\documentclass{amsart}
\usepackage{pgf,pgfarrows,pgfnodes,pgfautomata,pgfheaps,pgfshade}
\usepackage{
graphicx, amsfonts, amssymb, amsbsy, amsmath, amsthm, url}

\theoremstyle{plain}
\newtheorem{theorem}{Theorem}[section]

\newtheorem{proposition}[theorem]{Proposition}

\theoremstyle{definition}
\newtheorem{definition}[theorem]{Definition}
\newtheorem{example}[theorem]{Example}

\theoremstyle{remark}
\newtheorem{remark}{Remark}

\begin{document}

\title[Algebraic characterization of regular fractions]{Algebraic characterization of regular fractions under level permutations}

\author{Fabio Rapallo}
\address{Department of Science and Technological Innovation, University of Piemonte Orientale, Alessandria, Italy}
\email{fabio.rapallo@uniupo.it}
\author{Maria Piera Rogantin}
\address{Department of Mathematics, University of Genova, Genova, Italy}
\email{rogantin@unige.it}

\begin{abstract}
In this paper we study the behavior of the fractions of a factorial design under permutations of the factor levels. We focus on the notion of regular fraction and we introduce methods to check whether a given symmetric orthogonal array can or can not be transformed into a regular fraction by means of suitable permutations of the factor levels. The proposed techniques take advantage of the complex coding of the factor levels and of some tools from polynomial algebra. Several examples are described, mainly involving factors with five levels.
\end{abstract}


\maketitle

\section{Introduction}

In Design of Experiments the use of fractions of a full factorial design plays an important role when the observation of the response variable at each of the possible level combinations of the factors is impracticable. In general, the selected level combinations must satisfy optimality criteria to conveniently measure the impact of factors and their interactions on the mean and on the variability of the response variable.
In this framework, orthogonal arrays and regular fractions are widely used and both are based on properties of orthogonality among factors. The notion of orthogonality has  two main interpretations that coincide in the two level case: vector orthogonality and combinatorial orthogonality. Vector orthogonality allows one to construct linear models with non correlated factor effect estimators, but this concept is relevant with quantitative factors. Combinatorial orthogonality easily applies both in qualitative and quantitative cases. The complex coding of factor levels, extensively studied in \cite{pistone|rogantin:08} for asymmetric and multilevel designs, get together vector and combinatorial interpretations.

In this paper we consider how the orthogonal properties of a fraction change in presence of permutations of the factor levels, and in particular
we give some methods to check if a given fraction with qualitative factors is or not isomorphic (or equivalent) to a regular fraction up to permutations of the factor levels. Regular fractions are special orthogonal arrays where the factors and their interactions are either orthogonal or totally confounded.

Two fractional factorial designs are said isomorphic if one of them can be obtained from the other by reordering the runs, relabeling the
factors and/or permuting the levels of one or more factors. If a factor is quantitative, only the reverse permutation is allowed. Usually, the isomorphism is referred as geometric when the factors are quantitative, and as combinatoric when the factors are qualitative. Mixed situations can occur in practice, see \cite{katsaounis:12}.

Different methods to check the isomorphism between fractions has been done in literature. \cite{clark|dean:01}, \cite{katsaounis|dean:08}, and \cite{katsaounis:12} give necessary and sufficient conditions for isomorphism, based on the Hamming distance between two fractions, while \cite{katsaounis:13} uses the singular value decomposition of the design matrix for binary designs. \cite{ma|fang|lin:01} and \cite{pang|liu:11} consider an approach using the centered $L_2$-discrepancy for geometric isomorphism. In the latter paper, an algorithm  with low complexity is presented. \cite{gromping:16} introduce a different definition of regularity based on the canonical correlation to analyze the problem. \cite{tang|xu:14} essentially deal with an inverse problem: to find factor level permutations of a regular fraction to reduce contamination of non negligible interactions on the estimation of linear effects without increasing the run size.

We emphasize that most of the algorithms in literature are specific for binary factors, and even when defined in general for multilevel factors they are exemplified in the three levels case. Symmetric qualitative designs with three factor levels can be partitioned into isomorphic to a regular and not isomorphic to a regular fraction, while with larger number of levels the situation is more difficult to analyze. In fact, up to the three levels case, the isomorphism with respect to a regular fraction can be detected directly using the complex coding of levels and the indicator function, as argued later in Section \ref{poly-sect}.

A necessary condition for isomorphism between two fractions is that they have the same Generalized Word Length Pattern (GWLP) and this condition does not depend on the level coding, see \cite{wu|hamada:00}. \cite{fontana|rapallo|rogantin:16} derive a formula, computationally easy and of clear interpretation, for computing the GWLP of a fraction, based on the mean aberration.

In this paper we use tools from  Algebraic statistics, and in particular the polynomial indicator function to specify a fraction and its orthogonality/regularity properties, the polynomial representation of the permutations of the factor levels, and the representation through Latin squares of orthogonal arrays of strength two. In this framework, the coding of the $s$ levels of a factor by the $s$-th roots of the unity is essential. The methodology introduced here to detect isomorphism between fractions applies to symmetric designs with $s$ prime and the given examples consider $s\ge 5$.

The paper is organized as follows. In Section \ref{sec:alg} some relevant results of the algebraic theory of fractional factorial designs are presented. In Section \ref{poly-sect} the polynomial representation of factor level permutations is analyzed. Using polynomial conditions on the coefficients of such a representation, an algorithm to check if two fractions are isomorphic is given and some examples, implemented in the symbolic software CoCoA-5, are shown. In particular fractions of a $5^3$ design are checked to be isomorphic or not to a regular one. In Section \ref{ls-sect} the investigation if a fraction is isomorphic to a regular one or not is approached using the Latin square representation of orthogonal arrays of strength 2. When the isomorphism exists, the relevant level permutations and the generating equations of the regular fraction are recovered, by exploiting the properties of the complex coding of the levels. Such a check is based on the nullity of all the $2 \times 2$ minors of the multi-layer Latin squares in the numeric complex field. Several examples are discussed to show how to actually apply the proposed method. Finally, some further developments are illustrated in Section \ref{future-sect}.

\section{Algebraic characterization of fractional designs} \label{sec:alg}

In this section we present some relevant results of the algebraic theory of fractional designs. The interested reader can find further information, including the proofs of the propositions, in \cite{fontana|pistone|rogantin:00}, \cite{pistone|rogantin:08}, \cite{fontana|rapallo|rogantin:16}.

Let us consider an experiment with $m$ factors observed at $s$ levels each with $s$ a prime number.

Let us code the $s$ levels of  the $s$-th roots of the unity  $\omega_k=\exp\left(\sqrt{-1}\:  \frac {2\pi}{s} \ k\right)$, $k=0,\ldots,s-1$. We denote the level set by $\Omega_{s}=\left\{\omega_0,\ldots,\omega_{{s}-1}\right\}$.

As $\alpha=\beta \mod s$ implies $\omega_k^\alpha = \omega_k^\beta$, it is useful to introduce the residue class ring $\mathbb Z_s$ and the notation $[k]_s$ for the residue of $k \mod s$. For integer $\alpha$, we obtain $(\omega_k)^\alpha = \omega_{[\alpha k]_s}$. We also have $\omega_h\omega_k=\omega_{[h+k]_s}$. We drop the sub-$s$ notation when there is no ambiguity.

We denote by ${\mathcal D}$  the full factorial design with complex coding: $\mathcal D = \Omega_s^m$;
the cardinality of the full factorial design is $\# {\mathcal D}=s^m$.

We denote by $L$ the exponent set of the complex coded design: $L  =\mathbb Z_s^m$.
Notice that $L$ is both  the  exponent set of the complex coded design and the integer coded design. The elements of $L$ are denoted in vector notation by $\alpha$, $\beta, \ldots$:
\begin{equation*}
L = \left\{ \alpha= (\alpha_1,\ldots,\alpha_m) : \alpha_j = 0,\ldots,s-1, j=1,\ldots,m\right\} \ ;
\end{equation*}
$[\alpha-\beta]$ is the $m$-tuple $\left(\left[\alpha_1-\beta_1 \right], \ldots, \left[\alpha_m - \beta_m\right] \right)$.

In order to use polynomials to represent all the functions defined over ${\mathcal D}$, we define
\begin{itemize}
\item [-] $X_j$, the $j$-th component function, which maps a point of ${\mathcal D}$ to its $j$-th component,
$
X_j \colon {\mathcal D} \ni (\zeta_1,\ldots,\zeta_m)\ \longmapsto \ \zeta_j
$.
The function $X_j$ is a \emph{simple term} or, by abuse of terminology, a \emph{factor}.
\item [-] $X^\alpha=X_1^{\alpha_1} \cdot \ldots \cdot X_m^{\alpha_m}$, $\alpha \in L =  {\mathbb Z}_s^m$ i.e., the monomial function
$
X^\alpha : {\mathcal D} \ni (\zeta_1,\ldots,\zeta_m)\ \mapsto \ \zeta_1^{\alpha_1}\cdot \ldots \cdot \zeta_m^{\alpha_m}
$.
The function $X^\alpha$ is an \emph{interaction term}. As $s$ is a prime number, the interaction $X^\alpha$ takes values in $\Omega_s$.
\end{itemize}

The set of monomials $\{X^\alpha: \alpha \in L\}$ is an orthonormal basis of all the complex functions defined over ${\mathcal D}$.

Since we will make use occasionally factors with a non-prime number of levels, the behavior of the factors and of the interactions is summarized below:
\begin{itemize}
 \item [-]
Let $X_i$ be a simple term with level set $\Omega_s$.  Let us define $h=s/\text{gcd}(r,s)$ and let $\Omega_h$ be the set
of the $h$-th roots of the unity. The term $X_i^r$ takes all the values of $\Omega_h$ equally often over ${\mathcal D}$.

 \item [-]
Let $X^\alpha=X_{j_1}^{\alpha_{j_1}} \cdots X_{j_k}^{\alpha_{j_k}}$ be an interaction term of order $k$ where
$X_{j_i}^{\alpha_{j_i}}$ takes values in $\Omega_{h_{j_i}} $. Let us define $h=\text{lcm}\{h_{j_1}, \ldots, h_{j_k}\}$.
The interaction $X^\alpha$ takes values in $\Omega_h$ equally often  over ${\mathcal D}$.
\end{itemize}

Let $\mathcal F$ be a subset of the full factorial design $\mathcal D$. We consider here only fractions without replications.

\begin{definition} \label{de:indicator}
The \emph{indicator function} $F$ of a fraction ${\mathcal F}$ is a complex polynomial defined over ${\mathcal D}$ such that for each $\zeta \in {\mathcal D}$, $F(\zeta)$ is equal to $1$ if $\zeta \in \mathcal F$ and equal to $0$ if $\zeta \in \mathcal D \setminus \mathcal F$ .
We denote by $b_\alpha$ the coefficients of the representation of $F$  on ${\mathcal D}$ using the monomial basis
$\{X^\alpha, \ \alpha \in L\}$:
\begin{equation} \label{indf}
F(\zeta) = \sum_{\alpha \in L} b_\alpha X^\alpha(\zeta), \;\zeta\in{\mathcal D}, \;  b_\alpha \in \mathbb C \ .
\end{equation}
\end{definition}

\begin{proposition} \label{pr:ceqn}

Let ${\mathcal F}$ be a fraction with indicator function $F$.
\begin{enumerate}
\item The coefficients $b_\alpha$ of  $F$  are given by:
\begin{equation} \label{calpha}
b_\alpha = \frac 1 {\#{\mathcal D}} \sum_{\zeta \in {\mathcal F}} X^{[-\alpha]}(\zeta) =
 \frac{1}{\#{\mathcal D}} \sum_{h=0}^{s-1} n_{\alpha,s-h} \ \omega_{h}
\end{equation}
where  $n_{\alpha,h}$ is the number of the occurrences of  $\omega_h$ in $\{X^{\alpha}(\zeta):\zeta \in {\mathcal F}\}$.

In particular $b_0= {\#\mathcal F}/{\#{\mathcal D}}$.

 \item The term $X^\alpha$ is centered on $\mathcal F$ if, and only if, $b_\alpha=b_{[-\alpha]}=0$.
 \item The terms $X^\alpha$ and $X^\beta$
are orthogonal on $\mathcal F$ if, and only if, $b_{[\alpha-\beta]}=0$;
 \item If $X^\alpha$ is centered then, for each $\beta$ and $\gamma$ such that $\alpha=[\beta-\gamma]$ or
 $\alpha=[\gamma-\beta]$, $X^\beta$ is orthogonal to $X^\gamma$.
\item \label{prime} Let $s$ be prime. Then, the term $X^\alpha$ is centered on $\mathcal F$ if, and only if, its $s$
levels appear equally often: $n_{\alpha,0} = \cdots = n_{\alpha,s-1}$.
\end{enumerate}
\end{proposition}

As stated in the proposition above, the coefficients $b_\alpha$ encode interesting properties of the fraction such as
orthogonality among factors and interactions, and get together the combinatorial and vectorial orthogonality.

A common choice to select an experiment is to use orthogonal array. The assumption that interactions above a specified order are not present is translated into a combinatorial property on the frequency of the levels in the fraction. Let us denote with $\text{OA}(n,s^m,t)$ a symmetric  orthogonal array with $n$ rows and $m$ columns,  in which each column has $s$ symbols, and with strength $t$, as defined e.g. in \cite{wu|hamada:00}. Strength $t$ means that for every selection of $t$ columns of the matrix, all the elements of $\Omega_s^t$ appear equally often in the $t$ columns.

\begin{definition} \label{de:projectivity}
Let $I$ be a non-empty subset of $\{1,\ldots,m\}$. A fraction $\mathcal F$ {\em factorially projects} on the $I$-factors if the projection of ${\mathcal F}$ on the $I$-factors is a full factorial design where each point appears $k$ times. A fraction $\mathcal F$ is an {\em orthogonal array} of strength $t$ if it factorially projects on any $I$-factors with $\#I=t$.
\end{definition}

The proposition below shows a connection between the combinatorial definition of orthogonal array introduced above and the coefficients of the indicator function in Eq. \eqref{indf}.

\begin{proposition}  \label{pr:projectivity}
A fraction is an \emph{orthogonal array of strength $t$}  if, and
only if, all the coefficients of the indicator function up to the order $t$ are zero:
\[ b_{\alpha}=0 \quad \forall \ \alpha  \textrm{ of order up to }t, \ \alpha \ne (0, \ldots ,0)\ .
\]
\end{proposition}

\begin{definition} \label{wlp1}
Given an interaction $X^\alpha$ defined on a fraction ${\mathcal F}$ of the full factorial design ${\mathcal D}$, its \emph{aberration}, or degree of aliasing, $a_\alpha$ is given by the real number
\begin{equation*}
a_\alpha=  \frac{ \|b_{\alpha}\|_2^2 }{b_{0}^2}
\end{equation*}
where $\| x \|_2^2$ is the square of the Euclidean norm of the complex number $x$.

The GWLP $A({\mathcal F})=(A_1({\mathcal F}),\ldots,A_m({\mathcal F}))$ of a fraction ${\mathcal F}$  is defined as
\begin{equation*}
A_j({\mathcal F})= \sum_{\|\alpha \|_0 =j}  a_{\alpha} \quad j=1,\ldots,m \ ,
\end{equation*}
where $\|\alpha \|_0$ is the number of non-null elements of $\alpha$, i.e., the order of interaction of $X^\alpha$.
\end{definition}

The following proposition allows us to compute all the aberrations $a_\alpha$ without using complex numbers.
\begin{proposition}
Let $X^\alpha$ be a simple or interaction term with values in $\Omega_{t}$. Its aberration $a_\alpha$ is
\begin{equation*}
a_\alpha=\frac{1}{n^2}\left( \sum_{k=0}^{t-1}  \cos\left(\frac{2\pi}{t} k \right) \sum_{i=0}^{t-1} n_{\alpha,i}
n_{\alpha,[i-k]}  \right) \, .
\end{equation*}
\end{proposition}

Regular fractions are a subset of the orthogonal arrays. In a regular fraction any two
simple or interaction terms are either orthogonal or totally confounded.

\begin{definition} \label{def:reg}
A fraction $\mathcal F$ is \emph{regular} if there exist a sup-group $\mathcal L$ of $L$,
a group homomorphism $e$ from $\mathcal L$ to $\Omega_s$ such that the set of equations
\begin{equation} \label{eq:def-eq}
\{X^\alpha = \omega_{e(\alpha)} :   \alpha \in \mathcal L\}
\end{equation}
defines the fraction $\mathcal F$. If ${\mathcal H}$ is a minimal generator of the group ${\mathcal L}$, the set of equations $\{X^\alpha = {\omega_{e(\alpha)}} : \alpha \in {\mathcal H}\}$ is called the set of \emph{defining equations} of $\mathcal F$.
\end{definition}

\begin{proposition} \label{th:reg}
Let $\mathcal F$ be a fraction. The following statements are equivalent:
\begin{enumerate}
 \item \label{it:reg} The fraction $\mathcal F$ is regular according to Definition \ref{def:reg}.
  \item \label{it:reg-F} The indicator function of the fraction has the form
\begin{equation*}
F(\zeta)=\frac 1 {\#{\mathcal L}} \sum_{\alpha \in \mathcal L} \overline{\omega_{e(\alpha)}}\ X^\alpha(\zeta) \qquad \zeta \in \mathcal D \ .
\end{equation*}
where $\mathcal L$ is a given subset of $L$ and $e: \mathcal L \to \Omega_s$ is a given mapping.

 \item \label{it:reg-con} For each $\alpha, \beta \in  L$ the interactions represented on $\mathcal F$ by the
 terms $X^\alpha$  and $X^\beta$ are either orthogonal or totally aliased.
\end{enumerate}
\end{proposition}

Finally, we recall two basic definitions of isomorphic fractions. For details see e.g. \cite{deanetal:15}.

\begin{definition}
Two fractions are
\begin{itemize}
\item[-] \emph{combinatorially isomorphic} if one can be obtained from the other by reordering the runs, relabeling the factors and/or {\emph switching} the levels of one or more factors.

\item[-] \emph{geometrically isomorphic} if one design can be obtained from the other by reordering the runs, relabeling the factors and/or {\emph reversing} the level order of one or more factors
\end{itemize}
\end{definition}

The combinatorial isomorphism pertains to qualitative factors, while geometric isomorphism pertains to quantitative factors. In this paper we focus on qualitative factors mainly.

From the definition of the indicator function, it follows immediately that a reordering of the runs does not affect the polynomial representation of the indicator function. Moreover, a relabeling of the factors simply permutes the subscripts. From the expression of the indicator function, it is relatively simple to find the relevant relabeling. Therefore, the most interesting task in analyzing the isomorphism of fractions is to study the behavior of the fractions under permutations of the factor levels.

\section{Polynomial representation of the factor level permutations} \label{poly-sect}

In this section first we give an account of the polynomial representation of the level permutations for a single factor, then we extend this representation to several factors, and finally we use such characterization to check the combinatorial isomorphism of two fractions on some examples.

A level permutation is a function from $\Omega_s$ to $\Omega_s$ and it always admits a polynomial representation. In special cases, such polynomial reduces to a monomial, and therefore a permutation of the levels does not affect the regularity of a fraction.

\begin{proposition} \label{pr:recod}
A regular fraction is transformed into a regular fraction by the group of transformations generated by the following level permutations on the factor $X_j$:
\begin{enumerate}
\item \label{pr:recod1} Cyclical permutations $X_j  \rightarrow \omega_k X_j$ with $k=0,\dots, s-1$.
\item \label{pr:recod2} Power permutations, for $s$ prime, $X_j \rightarrow \omega_k X_j^h$  with $h=1,\dots,s-1$.
\end{enumerate}
Permutations of type $1.$ and $2.$, produce $s(s-1)$ permutations on the factor $X_j$, and produce the following transformed monomial on the term $X^{\alpha}$:
\begin{equation}\label{eq:perm}
\prod_{j=1}^m
\omega_{[\alpha_j k_j]}X_j^{[\alpha_j h_j]}.
\end{equation}
\end{proposition}
The proof of Prop. \ref{pr:recod} is in \cite{pistone|rogantin:08}. We observe only that under monomial permutations the absolute values of the indicator function coefficients do not change, so that a regular fraction is transformed into a regular fraction.

We highlight again that, for factors with two or three levels, all the level permutations have a monomial representation.

\begin{remark}
From Prop. \ref{pr:recod} the monomial representation of a geometric isomorphism follows. In fact, for quantitative factors, the two admissible level permutations are both in monomial form: $Y=X$ (the identity) and $Y=\omega_{s-1}X^{s-1}$ (the reversing of the factor levels).
\end{remark}

In the reminder of this section we characterize the polynomial representation of the permutations for general multilevel factors. Such characterization leads to a criterion to actually check if a given fraction may be or may be not transformed into a regular fraction after permutation of the levels of one or more factors.

Let $X$ be a factor with $s$ levels and let $\pi$ be a permutation of the level set $\Omega_s$. We denote by $Y$ the transformed factor, $Y=\pi(X)$. The polynomial representation of $Y$ is
\begin{equation}
Y = \sum_{h=0}^{s-1} u_h X^h \qquad u_h \in {\mathbb C}
\end{equation}
with $\pi(\omega_i)=\sum_{h=0}^{s-1} u_h \omega_i^h$.
The $u_h$'s  coefficients  are the solutions of the linear system
\begin{equation} \label{trans}
\left(\begin{array}{cccc} \omega_0^0 & \omega_1^0 & \ldots & \omega_{s-1}^0 \\
\omega^0_1 & \omega_1^1 & \ldots & \omega^{s-1}_1 \\
\vdots &  \vdots & & \vdots \\
\omega^0_{s-1} & \omega^1_{s-1} & \ldots & \omega^{s-1}_{s-1}
\end{array}\right) \left(\begin{array}{c} u_0 \\
u_1 \\
\vdots \\
u_{s-1}
\end{array}\right) = \left(\begin{array}{c} \pi(\omega_0) \\
\pi(\omega_1) \\
\vdots \\
\pi(\omega_{s-1})
\end{array}\right) \, .
\end{equation}
The matrix appearing in Eq. \ref{trans} is a Vandermonde matrix $V$. If we denote by $v_{h+1,k+1}$, with $h,k=0,\dots,s-1$,
the generic element of $V$, it is known that the inverse $V^{-1}$ of $V$ has generic element $v_{h+1,k+1}^{-1}/s$. In our case, from the results in Section \ref{sec:alg}:
\[
v_{h+1,k+1} = \omega_{h}^{k} = \omega_{[hk]} \quad \textrm{ and } \quad v_{h+1,k+1}^{-1} = \frac 1 s \omega_{[s-hk]}
\]
and the resulting system is
\begin{equation} \label{transinv}
\left(\begin{array}{c} u_0 \\
u_1 \\
\vdots \\
u_{s-1}
\end{array}\right) =  \frac 1 s
\left(\begin{array}{cccc} \omega_0 & \omega_0 & \ldots & \omega_0 \\
\omega_0 & \omega_{s-1} & \ldots & \omega_{1} \\
\vdots &  \vdots & & \vdots \\
\omega_0 & \omega_1 & \ldots & \omega_{s-1}
\end{array}\right) \left(\begin{array}{c} \pi(\omega_0) \\
\pi(\omega_1) \\
\vdots \\
\pi(\omega_{s-1})
\end{array}\right) \, .
\end{equation}

Full details on the Vandermonde matrices for the roots of the unity, their properties and applications to complex interpolations can be found in, e.g., \cite{corless|fillion:13}.

Combining the expression of $V^{-1}$ with the fact that $\pi$ is a permutation,we obtain constraints on the coefficients $u_0, \ldots, u_{s-1}$ as in the following proposition.

\begin{proposition} \label{prima}
 The coefficients $u_i$'s must satisfy the following equations:
\begin{itemize}
\item[(i)] $u_0 = 0$;

\item[(ii)] for all $q=2, \ldots, {s-1}$,
\begin{equation} \label{const-pow}
\sum_{h_1,\ldots,h_{s-1}=0}^{s-1} u_{h_1}\dots u_{h_{q-1}}u_{[-h_1 \ldots -h_{q-1}]} = 0 \, ;
\end{equation}

\item[(iii)] given a permutation $\pi$, we have $\sum_{h=1}^{s-1} u_h = \pi(\omega_0)$, and therefore

\item[(iv)] $(\sum_{h=1}^{s-1} u_h)^s-1 = 0$\, .
\end{itemize}
\end{proposition}

\begin{proof}
Let $\pi$ be a permutation of $\Omega_s$. For item $(i)$, observe that the first row in the system (\ref{transinv}) leads to
\[
u_0 = \frac 1 p  \sum_{i=0}^{s-1} \pi(\omega_i) =  \frac 1 p \sum_{i=0}^{s-1} \omega_i = 0 \, .
\]
The constraints in item $(ii)$ are derived in the same way, but using the powers $Y^2, \ldots ,Y^{s-1}$. For $Y^2$, the term of degree zero is $\sum_{h=0}^{s-1} u_iu_{s-i}$ which is the left-hand side in Eq. (\ref{const-pow}) for $q=2$ and in the same way one writes the corresponding degree zero terms for $q=3, \ldots, s-1$. Now, if $s$ is prime, all the powers $Y^2, \ldots, Y^{s-1}$ are permutations of $\Omega_s$ and the result follows from item $(i)$. For general $s$ (not prime necessarily), define $m=\gcd(q,s)$ and note that $Y^q$ contains $m$ times all the elements of the set $\Omega_{[s/m]}$, whose sum is again zero.
For item $(iii)$ it is enough to write
\[
\sum_{h=1}^{s-1} u_h = \sum_{h=0}^{s-1} \frac 1 s \sum_{k=0}^{s-1} \omega_k^{-h} \pi(\omega_k) =  \frac 1 s \sum_{k=0}^{s-1} \pi(\omega_k) \sum_{h=0}^{s-1} \omega_{[-kh]} =  \]
where the inner sum is always equal to zero except for $k=0$, and thus
\[
= \frac 1 s \left( s \pi(\omega_0) \right) = \pi(\omega_0) \, .
\]
Finally, item $(iv)$ follows from $(iii)$ by noting that the value of $\pi(\omega_0)$ may take any values in $\Omega_s$.
\end{proof}

It is interesting to write explicitly the equations in items $(i)$, $(ii)$, and $(iv)$ for the first values of $s$.

For $s=2$, we have $u_0=0$ and $u_1^2-1=0$, and such two equations characterize the only two possible permutations of $\Omega_2$. The same holds for $s=3$, where we obtain
\[
u_0=0 \, ; \qquad \qquad u_1u_2 = 0 \, ; \qquad \qquad (u_1+u_2)^3-1=0\, .
\]
From the second equation, we conclude that one among $u_1$ and $u_2$ is zero, providing an alternative proof to the fact that all the level permutations have a monomial representation for factors with two or three levels.

We illustrate now an example with $s=4$, i.e., a non-prime $s$. The conditions are:
\[ u_0=0 \, ; \qquad \qquad u_2^2+2u_1u_3 = 0 \, ;
\]
\[ u_2(u_1^2+u_3^2) = 0 \, ;  \qquad \qquad (u_1+u_2+u_3)^4-1=0\, . \]
When $s=4$ not all the monomial maps of the form $Y=\omega_hX^k$, $h=0,\ldots, 3$, $k=1,\ldots, 3$ are the polynomial representation of a permutation. Take for example the monomial map $Y=X^2$. This correspond to the transformation with coefficients $u_0=u_1=u_3=0$ and $u_2=1$. This is not a solution of the above equations, since the second equation is not satisfied.

When $s$ increases the situation becomes computationally less simple, since from Prop. \ref{prima} there are $s$ polynomial equations with degrees $1, \ldots, s$. Therefore, the degree of the polynomial system is $s!$, which is exactly the number of the permutations of $\Omega_s$.

To check if the system has a finite number of solutions one can apply the finiteness theorem (see Appendix A). It is enough to compute a Gr\"obner basis of the ideal generated by the $s$ equations in Prop. \ref{prima} and check if all the terms $u_0^{c_0}, \ldots , u_{s-1}^{c_{s-1}}$ are all leading terms of polynomials in the Gr\"obner basis for some exponents $c_0, \ldots, c_s$.

\begin{example}
For $s=5$ the system in Eq. \eqref{transinv} yields $5$ equations, and the Gr\"obner basis of the corresponding polynomial ideal is formed by $28$ polynomials. Among them, the five polynomials displayed below have leading term of the form \verb"u[i]^c[i]" for appropriate exponents \verb"c[i]" for all \verb"i", and therefore  the finiteness theorem in Appendix A applies. Then the polynomial system has a finite number of solutions.

{\footnotesize\begin{verbatim}
u[0],
u[1]^5 +u[2]^5 +u[3]^5 +(-20)*u[1]^3*u[3]*u[4] +
       +(-20)*u[1]*u[2]*u[4]^3 +u[4]^5 +(-1),
u[2]^6 +(23)*u[1]^4*u[4]^2 +(16)*u[1]^2*u[3]*u[4]^3 +
       +(19)*u[3]^2*u[4]^4 +(20)*u[2]*u[4]^5 +(-1)*u[2],
u[3]^6 +(48)*u[1]*u[3]^3*u[4]^2 +(26)*u[2]^3*u[4]^3 +
       +(-81)*u[1]^2*u[4]^4 +(27)*u[3]*u[4]^5 +(-1)*u[3],
u[4]^11 +(-12628/625)*u[1]^3*u[3]*u[4]^2 +(-77/625)*u[1]*u[3]^2*u[4]^3 +
        +(-12639/625)*u[1]*u[2]*u[4]^4 +(-121/625)*u[4]^6 +(-504/625)*u[4]
\end{verbatim}}
\end{example}

Finally, an interesting property of the coefficients $u_i$'s concern their expression in terms of the roots of the unity.

\begin{proposition}
Up to the constant $1/s$, the coefficients $u_i$ are integer non-negative combinations of the $s$-th roots of the unity:
\begin{equation*}
u_h = \sum_{r=0}^{s-1} v_r \omega_{r} \, ,  \qquad v_{r} \in {\mathbb N}
\end{equation*}
When the number of levels $s$ is prime, for all permutations $\pi$, such representation of the coefficients is unique up to an additive integer constant.
\end{proposition}
\begin{proof}
The first part follows directly from the expression of the inverse of the Vandermonde matrix in Eq. \eqref{transinv}. For the uniqueness, see \cite{pistone|rogantin:08}.
\end{proof}

Now we show how to use the equations above in order to study the isomorphism between two fractions, by merging polynomial constraints on the support of a full factorial design and the polynomial constraints in Prop. \ref{prima}. The computations below are carried out with the free software {\tt CoCoA}, see \cite{CoCoA-5}. Also in these examples we will make use of basic tools from Computational Commutative Algebra, such as polynomial ideal, Gr\"obner basis, Normal Form. The basic definitions and results needed here are collected in Appendix A. Useful techniques to handle polynomials can be found in \cite{abbott:02}.

We write the polynomial indicator function of the two fractions under investigation, and we do some algebraic manipulations in order to obtain the coefficients of the (possible) permutation needed to transform the first fraction into the second one. In particular, we check if a fraction is isomorphic to a regular one. Our examples are given in the $5^3$ case, where there is  only one regular orthogonal array with strength $2$ up to monomial transformations. There are several online databases of orthogonal arrays. The examples analyzed here are taken from \cite{eendebak:sito}, generated with the algorithm introduced in \cite{schoen:10}.

Let $\mathcal F_0$ and $\mathcal F_1$ be the two fractions to compare, with polynomial indicator functions $F_0$ and $F_1$ respectively.
Consider  a generic transformation $\pi=(\pi_1,\pi_2,\pi_3)$ where $\pi_j$ acts on the factor $X_j$:
\[
\pi_j: \ X_j \longrightarrow \sum_{k=0}^4 u_{kj}X_j^k \, .
\]
In particular, if we want to check the regularity of the fraction $\mathcal F_1$, then $F_0$ is the indicator function of the
regular fraction with defining equation $X_1X_2X_3=\omega_0$, namely $F_0^{(r)} = \frac 1 4 \sum_{k=0}^4 \left(X_1 \ X_2 \ X_3 \right)^k $

\begin{enumerate}
\item Consider  the ring of the indeterminates
\begin{itemize}
  \item[-] \texttt{x[1],x[2],x[3]}, the factors;
  \item[-] \texttt{u[0..4, 1..3]}, the $5\times 3$ transformation coefficients;
  \item[-] \texttt{w}, the $5$-th primitive root of the unity, satisfying the equation \\
  \verb"1+w+w^2+w^3+w^4=0". The indeterminate \texttt{w} is considered here as a parameter.
\end{itemize}

\item Input \verb"F0" and \verb"F1", the indicator functions of the two fractions to be compared (minus one).

\item Consider \texttt{I}, the ideal  generated by the polynomials defining the full factorial design ($x_i^5-1$, for $i=1,2,3$) and the $3 \times 5$ polynomials with the conditions for the transformation coefficients. The CoCoA code for such polynomials is:

    {\footnotesize{\begin{verbatim}
L:=NewList(5);
L[1] := [u[0,j] |j in 1..3];
L[2] := [(Sum([u[i,j] | i In 1..4]))^5 -1 |j in 1..3];
L[3] := [Sum([u[i,j]*u[Mod(-i,5),j]| i In 1..4]) |j in 1..3];
L[4] := [Sum(Sum([[u[i,j]*u[h,j]*u[Mod(-i-h,5),j]| i In 1..4]
                                        | h In 1..4])) |j in 1..3];
L[5] := [Sum(Sum(Sum([[[u[i,j]*u[h,j]*u[k,j]*u[Mod(-i-h-k,5),j]
                |i In 1..4]| h In 1..4]| k In 1..4]))) |j in 1..3];
    \end{verbatim}}}
\item Compute \verb"NF_P_F0", the normal form of the transformation of $F_0$ by $\pi$ in the quotient space $K/I$ and compute \verb"Coe_F0", the list of the coefficients of the terms in \texttt{x[1],x[2],x[3]} appearing in \verb"NF_P_F0".
\item Compute \verb"Coe_F1", the list of the coefficients of the terms in \texttt{x[1], x[2], x[3]} appearing in \texttt{F1}.
\item Compute \verb"Coe_Diff", the difference between the coefficients \verb"Coe_F1" and the coefficients \verb"Coe_F0", for all the terms in \verb"Coe_F0".
\item \label{it:stop} Compute the Gr\"obner Basis of the ideal generated by the polynomials in \verb"Coe_Diff" and the polynomials in \verb"L[0..4]" with the conditions for the transformation coefficients.
\end{enumerate}
If the  Gr\"obner basis is empty, then the two fractions are not isomorphic, otherwise the
 Gr\"obner Basis contains equations on the transformation coefficients that allow us to find the permutations.

Notice that, if $F_0^{(r)}$ is the indicator function with generating equation $X_1X_2X_3=1$, then the number of terms of \verb"NF_P_F0" is $6401$, while, obviously, the length of \verb"Coe_F0" is $65$, the number of the interactions of order $3$ plus the constant.

\begin{example}
Let $\mathcal F_A$, $\mathcal F_B$ and $\mathcal F_C$ be three fractions of a $5^3$ factorial design, whose indicator functions are shown in Appendix B. The first two fractions are listed in \cite{eendebak:sito}. The fraction $\mathcal F_B$ is a regular fraction evidently.

Using the previous algorithm we checked if they are isomorphic to the regular fraction $\mathcal F^{(r)}$ above. In the firs case, the Gr\"obner Basis of the step \ref{it:stop} has only the element $1$; then $\mathcal F_A$ is not isomorphic to any regular fraction of a $5^3$ factorial design. In the last case the Gr\"obner Basis  has been computed in $29$ secs. of CPU time and contains 91 elements. A solution is:
\begin{multline*}
u_{.,1}=(0,1,0,0,0)\qquad u_{.,3}=(0,0,0,0,1)\\
u_{.,2}=\frac 1 5 \left(0,2-\omega_2-\omega_3,\ 2\omega_1+\omega_2+2\omega_3, \ \omega_1+\omega_2-\omega_3-1, \ 2\omega_1-\omega_2-1\right) \end{multline*}
that corresponds to no permutation on the first factor, a power permutation on the third factor and the switch between the levels $0$ and $1$ on the second one.
\end{example}

\section{Regularity check of multi-level orthogonal arrays} \label{ls-sect}

In this section we approach the problem of checking the regularity of a multi-level orthogonal array using the complex coding of the factor levels. In particular we provide a result which leads us to check if a given orthogonal array may be regarded as a regular fraction under suitable permutations of the factor levels. This technique exploits the connections between orthogonal arrays and Latin squares, see for instance \cite{keedwell|denes:15} and \cite{hedayat:99}, and it is based on the generating equations of the regular fraction rather than on the whole indicator function.

We focus on orthogonal arrays with strength 2. Let ${\mathcal F}$ be an $OA(n, s^m, 2)$ with $s$ prime. A regular fraction with strength $2$ has at least one generating equation involving only $3$ factors. Therefore we first look at generating equations involving three factors. Without loss of generality, let us consider the factors $X_{1},X_{2},X_{3}$. Let
\begin{equation} \label{gen-eq}
X_1^{\alpha_1}X_2^{\alpha_2}X_3^{\alpha_3} = \omega_k
\end{equation}
be a generating equation of $\mathcal F$, with $\alpha_1,\alpha_2,\alpha_3 \in \{1, \ldots, s-1\}$ and $\omega_k \in \Omega_s$. For brevity, we say that $X_1, X_2, X_3$ form a generating equation of the orthogonal array ${\mathcal F}$.

If $X_3$ is a function of $X_1$ and $X_2$, we can consider the $s \times s$ table $C=X_3(X_1,X_2)$ containing the values of $X_3$ as a function of $X_1$ and $X_2$, i.e., $C_{j_1,j_2}=x_3$ given $x_1=\omega_{j_1}$ and $x_2=\omega_{j_2}$. Since the strength of the orthogonal array is $2$, the table $C$ may be regarded as a $s \times s$ Latin square with values in $\Omega_s$. The main result can be stated as follows.

\begin{theorem} \label{th-LS}
Let $X_{1},X_{2},X_{3}$ be three factors of an $OA(n, s^m, 2)$, $s$ prime. If $X_3$ is a function of $X_1$ and $X_2$, let $X_3(X_1,X_2)$ be the corresponding Latin square.
\begin{itemize}
\item[(a)] If $X_{1},X_{2},X_{3}$ form a generating equation, then $X_3(X_1,X_2)$ has rank $1$ in ${\mathbb C}$, i.e., all $2 \times 2$ minors of $X_3(X_1,X_2)$ vanish in ${\mathbb C}$;

\item[(b)] If there is a permutation $\pi_3$ of $\Omega_s$ such that $(\pi_3(X_3))(X_1,X_2)$ is a Latin square with rank $1$ in ${\mathbb C}$, then there exist permutations $\pi_1$ and $\pi_2$ such that $\pi_1(X_1), \pi_2(X_2), \pi_3(X_3)$ form a generating equation.
\end{itemize}
\end{theorem}
\begin{proof}
$(a)$ By hypothesis there exist $\alpha_1,\alpha_2,\alpha_3 \in \{1, \ldots, s-1\}$ and $\omega_k \in \Omega_s$ such that $X_1^{\alpha_1}X_2^{\alpha_2}X_3^{\alpha_3} = \omega_k$. Since $s$ is prime, we can assume $\alpha_3=1$. In fact, given a generating equation $X_1^{\alpha_1}X_2^{\alpha_2}X_3^{\alpha_3} = \omega_k$, there exists $r$ such that $[r\alpha_3]=1$ and the equation
\[
X_1^{[r\alpha_1]}X_2^{[r\alpha_2]}X_3 = \omega_{[rk]}
\]
is also a generating equation of the fraction. Consider the Latin square $C=X_3(X_1,X_2)$. The entry $C_{j_1,j_2}$ of $C$ is $C_{j_1,j_2}= \omega_{[rk]} \omega_{j_1}^{[-r\alpha_1]}\omega_{j_2}^{[-r\alpha_2]}$ for $j_1,j_2=0, \ldots, s-1$ and the generic $2 \times 2$ minor of $C$ is
\begin{multline*}
\omega_{[2rk]}\left( \omega_{[j_{11}-r\alpha_1]}  \omega_{[j_{21}-r\alpha_2]} \omega_{[j_{12}-r\alpha_1]} \omega_{[j_{22}-r\alpha_2]} - \right. \\ \left. \omega_{[j_{11}-r\alpha_1]}  \omega_{[j_{22}-r\alpha_2]} \omega_{[j_{12}-r\alpha_1]} \omega_{[j_{21}-r\alpha_2]} \right)
\end{multline*}
that equals $0$ for all pairs of distinct row indices $j_{11},j_{12} \in \{0, \ldots, s-1\}$ and for all pairs of distinct column indices $j_{21},j_{22} \in \{0, \ldots, s-1\}$.

$(b)$ Suppose that there is a permutation $\pi_3$ of $X_3$ such that the table $C=\pi_3(X_3)(X_1,X_2)$ is a Latin square with all $2 \times 2$ minors equal to $0$. Then apply suitable permutations $\pi_1$ and $\pi_2$ to $X_1$ and $X_2$, respectively, in order to obtain a Latin square in reduced form, i.e., with the first row and column lexicographically ordered. Now it is immediate to check that $\pi_3(X_3)=\pi_1(X_1)\pi_2(X_2)$ and therefore a defining equation after the level permutations is $\pi_1(X_1)\pi_2(X_2)\pi_3(X_3)^{[-1]}=\omega_0$.
\end{proof}

Some remarks on part $(b)$ of the theorem above are now in order. We can exploit the monomial representation of the permutations in Prop. \ref{pr:recod} to reduce considerably the computational cost. First, observe that the permutations to be checked on $X_3$ are at most $(s-2)!$. In fact, we can exclude the powers (to each defining equation correspond other $(s-2)$ equivalent ones) and the $s$ cyclic permutations (they only affect the constant term). For instance, if $s=5$, there are $120$ level permutations but only six are to be checked. Secondly, the relevant permutations of $X_1$ and $X_2$ are the permutations needed to put the Latin square in reduced form. Additionally, if the permutations $\pi_1$ and $\pi_2$ can be expressed in monomial form (powers and/or cyclic permutations), then the defining equation can be written without actual permutations on $X_1$ and $X_2$.

Finally, note that the permutations of the factors are not uniquely defined: for instance, a shift of the form $\omega_hX$ can be applied to whatever factor and it generates a unique transformation in the generating equation.

Before analyzing the general case of symmetric multilevel designs, we present some  applications of the theorem above in the simple case of orthogonal arrays with 3 factors and strength 2, so that there is only one defining equation.

\begin{example} \label{ex3}
In the framework of the $5^3$ full factorial design, consider the three orthogonal arrays with strength $2$ identified by the three Latin squares in Figure \ref{figls1}. To ease the readability of the tables, we write $k$ in place of $\omega_k$. The three designs are built from the $2$ non-isomorphic orthogonal arrays with strength $2$ listed in \cite{eendebak:sito}. The array in (a) is the second fraction in \cite{eendebak:sito}, the array in (b) is obtained by the previous one, with some permutations on the factor levels, and the array in (c) is the first fraction in \cite{eendebak:sito}. The three indicator functions of these fractions are in Appendix B. In this case of fractions of the $5^3$ design with $25$ runs, it is known that there are no other non-isomorphic orthogonal arrays.

\begin{figure}
\begin{center}
\begin{tabular}{ccc}
\begin{tabular}{c|c|c|c|c|c|}
  \multicolumn{1} {c} { }    &
  \multicolumn{1} {c} {0} &  \multicolumn{1} {c} {1} &  \multicolumn{1} {c} {2} &  \multicolumn{1} {c} {3} &  \multicolumn{1} {c} {4}  \\
  \cline{2-6}
0 & 0 & 1 & 2 & 3 & 4    \\  \cline{2-6}
1 & 1 & 2 & 3 & 4 & 0  \\ \cline{2-6}
2 & 2 & 3 & 4 & 0 & 1    \\ \cline{2-6}
3 & 3 & 4 & 0 & 1 & 2    \\ \cline{2-6}
4 & 4 & 0 & 1 & 2 & 3     \\ \cline{2-6}
\end{tabular} &
\begin{tabular}{c|c|c|c|c|c|}
  \multicolumn{1} {c} { }    &
  \multicolumn{1} {c} {0} &  \multicolumn{1} {c} {1} &  \multicolumn{1} {c} {2} &  \multicolumn{1} {c} {3} &  \multicolumn{1} {c} {4}  \\
  \cline{2-6}
 0& 2 & 4 & 0 & 1 & 3   \\  \cline{2-6}
 1& 0 & 2 & 1 & 3 & 4  \\  \cline{2-6}
 2& 3 & 1 & 4 & 2 & 0   \\ \cline{2-6}
 3& 4 & 3 & 2 & 0 & 1    \\ \cline{2-6}
 4& 1 & 0 & 3 & 4 & 2  \\ \cline{2-6}
\end{tabular} &
\begin{tabular}{c|c|c|c|c|c|}
  \multicolumn{1} {c} { }    &
  \multicolumn{1} {c} {0} &  \multicolumn{1} {c} {1} &  \multicolumn{1} {c} {2} &  \multicolumn{1} {c} {3} &  \multicolumn{1} {c} {4}  \\
  \cline{2-6}
 0 & 0 & 1 & 2 & 3 & 4  \\  \cline{2-6}
 1 & 1 & 0 & 3 & 4 & 2  \\  \cline{2-6}
 2 & 2 & 3 & 4 & 0 & 1  \\  \cline{2-6}
 3 & 3 & 4 & 1 & 2 & 0  \\  \cline{2-6}
 4 & 4 & 2 & 0 & 1 & 3   \\  \cline{2-6}
\end{tabular} \\
 (a) &  (b) &  (c) \\
 \end{tabular}
 \end{center}
\caption{\label{figls1}The three orthogonal arrays of Example \ref{ex3}.}
\end{figure}

\begin{itemize}
\item[(a)] All the $2 \times 2$ minor vanish in ${\mathbb C}$. The defining equation of this fraction, without any permutations, is $X_1X_2=X_3$ or equivalently $X_1X_2X_3^4=\omega_0$.

\item[(b)] This fraction has been built from the previous one by applying the permutations $(\omega_3,\omega_2,\omega_4,\omega_1,\omega_0)$ to $X_1$, $(\omega_2,\omega_0,\omega_1,\omega_4,\omega_3)$ to $X_2$ and $(\omega_2,\omega_4,\omega_0,\omega_3,\omega_1)$ to $X_3$. Among the five possible (other than the identity) on $X_3$, we observe that the permutation $(\omega_4,\omega_3,\omega_0,\omega_2,\omega_1)$ produces the Latin square in Figure \ref{figls2}, where all the $2 \times 2$ minor vanish in ${\mathbb C}$.

    \begin{figure}
    \begin{center}
    \begin{tabular}{c|c|c|c|c|c|}
  \multicolumn{1} {c} { }    &
  \multicolumn{1} {c} {0} &  \multicolumn{1} {c} {1} &  \multicolumn{1} {c} {2} &  \multicolumn{1} {c} {3} &  \multicolumn{1} {c} {4} \\
  \cline{2-6}
0 & 0 & 1 & 4 & 3 & 2  \\  \cline{2-6}
1 & 4 & 0 & 3 & 2 & 1  \\ \cline{2-6}
2 & 2 & 3 & 1 & 0 & 4  \\ \cline{2-6}
3 & 1 & 2 & 0 & 4 & 3  \\ \cline{2-6}
4 & 3 & 4 & 2 & 1 & 0 \\ \cline{2-6}
\end{tabular}\end{center}
\caption{\label{figls2}The Latin square of the orthogonal array (b) of Example \ref{ex3} after the permutation $(\omega_4,\omega_3,\omega_0,\omega_2,\omega_1)$ on $X_3$.}
\end{figure}

Looking at the rows and columns beginning with $\omega_0$, we can read easily the permutations of the levels of $X_1$ and $X_2$ needed to obtain a Latin square in reduced form. From the first column we read that the permutation on $X_1$ is $(\omega_0, \omega_3, \omega_2, \omega_4, \omega_1)$, while in the first row we read that the permutation on $X_2$ is $(\omega_0, \omega_1, \omega_4, \omega_3, \omega_2)$. Notice that such permutations are those used in the construction of the orthogonal array up to a shift of four levels for $X_1$ and one level for $X_2$. In this example, the relevant permutations can not be expressed in monomial form, and thus the permutations of the factor levels are unavoidable.

\item[(c)] The Latin squares after the $5$ relevant permutations of $X_3$ are in Figure \ref{figls3}. The first minor in all tables is $\omega_0^2-\omega_1^2 = 1-\omega_2 \ne 0$ and this is enough to conclude that this fraction can not be transformed into a regular one.

\begin{figure}
\begin{center}
\begin{tabular}{ccc}
\begin{tabular}{c|c|c|c|c|c|}
  \multicolumn{1} {c} { }    &
  \multicolumn{1} {c} {0} &  \multicolumn{1} {c} {1} &  \multicolumn{1} {c} {2} &  \multicolumn{1} {c} {3} &  \multicolumn{1} {c} {4}  \\
  \cline{2-6}
0 & 0 & 1 & 3 & 2 & 4 \\ \cline{2-6}
1 & 1 & 0 & 2 & 4 & 3  \\ \cline{2-6}
2  &  3 & 2 & 4 & 0 & 1    \\ \cline{2-6}
3  &  2 & 4 & 1 & 3 & 0    \\ \cline{2-6}
4  &  4 & 3 & 0 & 1 & 2     \\ \cline{2-6}
\end{tabular} &
\begin{tabular}{c|c|c|c|c|c|}
  \multicolumn{1} {c} { }    &
  \multicolumn{1} {c} {0} &  \multicolumn{1} {c} {1} &  \multicolumn{1} {c} {2} &  \multicolumn{1} {c} {3} &  \multicolumn{1} {c} {4}  \\
  \cline{2-6}
0 &  0 & 1 & 3 & 4 & 2   \\  \cline{2-6}
 1 &  1 & 0 & 4 & 2 & 3  \\  \cline{2-6}
 2 &  3 & 4 & 2 & 0 & 1   \\ \cline{2-6}
 3 &  4 & 2 & 1 & 3 & 0    \\ \cline{2-6}
 4 &  2 & 3 & 0 & 1 & 4  \\ \cline{2-6}
\end{tabular} &
\begin{tabular}{c|c|c|c|c|c|}
  \multicolumn{1} {c} { }    &
  \multicolumn{1} {c} {0} &  \multicolumn{1} {c} {1} &  \multicolumn{1} {c} {2} &  \multicolumn{1} {c} {3} &  \multicolumn{1} {c} {4}  \\
  \cline{2-6}
 0 & 0 & 1 & 2 & 4 & 3  \\  \cline{2-6}
 1 & 1 & 0 & 4 & 3 & 2  \\  \cline{2-6}
 2 & 2 & 4 & 3 & 0 & 1   \\  \cline{2-6}
 3 & 4 & 3 & 1 & 2 & 0  \\  \cline{2-6}
 4 & 3 & 2 & 0 & 1 & 4    \\  \cline{2-6}
\end{tabular} \\
\end{tabular}

\medskip
\begin{tabular}{ccc}
\begin{tabular}{c|c|c|c|c|c|}
  \multicolumn{1} {c} { }    &
  \multicolumn{1} {c} {0} &  \multicolumn{1} {c} {1} &  \multicolumn{1} {c} {2} &  \multicolumn{1} {c} {3} &  \multicolumn{1} {c} {4}  \\
  \cline{2-6}
0 & 0 & 1 & 4 & 2 & 3 \\ \cline{2-6}
1 & 1 & 0 & 2 & 3 & 4  \\ \cline{2-6}
2 & 4 & 2 & 3 & 0 & 1    \\ \cline{2-6}
3 & 2 & 3 & 1 & 4 & 0    \\ \cline{2-6}
4 & 3 & 4 & 0 & 1 & 2     \\ \cline{2-6}
\end{tabular} &
\begin{tabular}{c|c|c|c|c|c|}
  \multicolumn{1} {c} { }    &
  \multicolumn{1} {c} {0} &  \multicolumn{1} {c} {1} &  \multicolumn{1} {c} {2} &  \multicolumn{1} {c} {3} &  \multicolumn{1} {c} {4}  \\
  \cline{2-6}
0 &  0 & 1 & 4 & 3 & 2    \\  \cline{2-6}
 1 &  1 & 0 & 3 & 2 & 4 \\  \cline{2-6}
2 &  4 & 3 & 2 & 0 & 1    \\ \cline{2-6}
  3 &  3 & 2 & 1 & 4 & 0    \\ \cline{2-6}
 4 &  2 & 4 & 0 & 1 & 3 \\ \cline{2-6}
\end{tabular}
\end{tabular}
\end{center}
\caption{\label{figls3}The five Latin squares obtained from the orthogonal array (c) of Example \ref{ex3} after the five non-identical relevant permutations.}
\end{figure}
\end{itemize}
\end{example}

\begin{remark}
In the part $(c)$ of the previous example, note that the $2 \times 2$ upper-left matrix has non-zero determinant for all permutations of the levels of $X_3$. In fact, if the permutation $\pi_3$ maps $\omega_0$ into $\omega_{j_0}$ and $\omega_1$ into $\omega_{j_1}$, with $j_0 \ne j_1$, one obtains the minor $\omega_{j_1}^2-\omega_{j_0}^2 = \omega_{[2j_1-2j_0]} \ne 0$. This remark can also be used to build non-regular orthogonal arrays also in case of a large number of levels, as illustrated in the example below.
\end{remark}

\begin{example}
The Latin square in Figure \ref{figls4} represents an orthogonal array of strength $2$ of the $7^3$ design. It has been defined starting from the upper-left $2 \times 2$ sub-matrix, and then completed in the remaining entries. By construction, it is a non-regular design even under permutations of the factor levels.

\begin{figure}
\begin{center}
\begin{tabular}{c|c|c|c|c|c|c|c|}
  \multicolumn{1} {c} { }    &
  \multicolumn{1} {c} {0} &  \multicolumn{1} {c} {1} &  \multicolumn{1} {c} {2} &  \multicolumn{1} {c} {3} &  \multicolumn{1} {c} {4} &  \multicolumn{1} {c} {5} &  \multicolumn{1} {c} {6} \\
  \cline{2-8}
0 & 0 &   1 &   2 &   3  &  4  &  5 &   6  \\  \cline{2-8}
1 & 1 &   0 &   4 &   5  &  2  &  6 &   3  \\ \cline{2-8}
2 & 4 &   2 &   3 &   6  &  5  &  0 &   1   \\ \cline{2-8}
3 & 6 &   3 &   5 &   0  &  1  &  4 &   2  \\ \cline{2-8}
4 & 5 &   4 &   1 &   2  &  6  &  3 &   0  \\ \cline{2-8}
5 & 3 &   5 &   6 &   1  &  0  &  2 &   4  \\ \cline{2-8}
6 & 2 &   6 &   0 &   4  &  3  &  1 &   5 \\ \cline{2-8}
\end{tabular}\end{center}
\caption{\label{figls4}A Latin square representing a non regular orthogonal array of strength 2 of the $7^3$ design.}
\end{figure}
\end{example}

Theorem \ref{th-LS} can be used for constructing an algorithm to check the regularity of orthogonal arrays with strength $2$ and with an arbitrary number of factors, under permutations of the factor levels.

First, consider an orthogonal array with one defining equation involving $m>3$ factors. Theorem \ref{th-LS} can be applied recursively layer by layer. For instance fix $m=4$. The factors $X_1, X_2, X_3, X_4$ form a generating equation if and only if
\begin{equation*}
X_1^{\alpha_1}X_2^{\alpha_2}X_3^{\alpha_3}X_4 = \omega_k
\end{equation*}
for some $\alpha_1,\alpha_2,\alpha_3 \in \{1, \ldots, s-1\}$ and $\omega_k \in \Omega_s$. Thus, for each $x_4=\omega_{j_4}$, $j_4=0, \ldots, s-1$, the equation
\begin{equation*}
X_1^{\alpha_1}X_2^{\alpha_2}X_3^{\alpha_3} = \omega_{[k-j_4]}
\end{equation*}
is satisfied.
Conversely, if Theorem \ref{th-LS} applies to each layer  $x_4 = \omega_{j_4}$, $j_4=0, \ldots, s-1$, with the same permutation, and there exist permutations $\pi_1,\pi_2,\pi_3$ on $X_1,X_2,X_3$ respectively such that
\begin{equation*}
\pi_1(X_1)^{\alpha_1}\pi_2(X_2)^{\alpha_2}\pi_3(X_3)^{\alpha_3} = \omega_{k(j_4)}
\end{equation*}
with different $\omega_{k(j_4)}$, for $j_4=0,\ldots, s-1$, then $k(j_4)$ defines a permutation $\pi_4$ for $X_4$ and
\begin{equation*}
\pi_1(X_1)^{\alpha_1}\pi_2(X_2)^{\alpha_2}\pi_3(X_3)^{\alpha_3}\pi_4(X_4)^{[s-1]} = \omega_0
\end{equation*}
is a defining equation for the orthogonal array.

For orthogonal arrays with strength $2$ and $m$ factors, we check the possible defining equations in the following order.
\begin{enumerate}
\item First, check all the $3$-tuples.

\item Then, if the defining equations with $3$ factors are not sufficient to define the orthogonal array, check the defining equations with $4$ or more factors.
\end{enumerate}
Notice that the number of (independent) defining equations is $\#{\mathcal D}/\#{\mathcal F}$.

The regularity check is based on the following property of regular fractions, that can be easily proved within the framework of the complex coding of the level factors. The proof of this result is based on the properties of the elimination ideals, see Appendix A for some basic definitions.

\begin{proposition}
Let ${\mathcal F}$ be a regular fraction of a $s^m$ design, and let $I \subset \{1, \ldots, m\}$. Denote with ${\mathcal F}_I$ the projection of ${\mathcal F}$ onto the $I$-factors. Apart from the multiplicity, ${\mathcal F}_I$ is either a full factorial design or a regular fraction.
\end{proposition}
\begin{proof}
Let ${\overline I} = \{1 , \ldots, m\} \setminus I$. Let us define the ideal ${\mathcal I}({\mathcal F})$ as the ideal in ${\mathbb C}[x_1, \ldots, x_m]$ generated by the binomials $x_j^s-1=0$, $j=1, \ldots, m$ and by the generating equations of ${\mathcal F}$. The ideal ${\mathcal I}({\mathcal F})$ is a binomial ideal. In fact, two factors or interactions are either orthogonal or totally confused and this yields only binomial equations. The ideal ${\mathcal I}({\mathcal F}_I)$ is the elimination ideal of ${\mathcal I}({\mathcal F})$ with respect to the variables $x_j$, $j \in {\overline I}$. From the results in Chapter 3 of \cite{cox|little|oshea:15}, ${\mathcal I}({\mathcal F}_I)$ is also a binomial ideal and two cases may arise: (a) the binomials $x_j^s-1=0$, $j \in I$ generate ${\mathcal I}({\mathcal F}_I)$, and this means that ${\mathcal F}_I$ is a full factorial design on the $I$ factors; (b) there are other generators. From the definition of ideal, such polynomials define ${\mathcal F}_I$ as a regular fraction.
\end{proof}

\begin{remark}  \label{rem-sempl}
To ease computations, remember that a defining equation with a given number of factors cannot include simultaneously all the factors of a defining equation with a lower number of factors.
\end{remark}

\begin{example} \label{ex55}
Consider the $5^{5-2}$ regular fraction defined by
\begin{equation} \label{def5a5}
X_1^2X_2X_3 =\omega_1\, \qquad \qquad X_1X_2X_4X_5 = \omega_1
\end{equation}
and a fractions obtained by permuting the levels of two factors. The indicator function of this fraction is written in Appendix B and has $289$ nonzero monomials. We apply our technique to the new fraction and we show how to recover the defining equations \eqref{def5a5} and the correct permutations starting from the fraction points.

First we check the defining equations with $3$ factors. We obtain a valid Latin square only with $X_1, X_2, X_3$ as in Figure \ref{figls5}.
\begin{figure}
    \begin{center}
    \begin{tabular}{c|c|c|c|c|c|}
  \multicolumn{1} {c} { }    &
  \multicolumn{1} {c} {0} &  \multicolumn{1} {c} {1} &  \multicolumn{1} {c} {2} &  \multicolumn{1} {c} {3} &  \multicolumn{1} {c} {4} \\
  \cline{2-6}
0 & 1 & 0 & 4 & 3 & 2  \\  \cline{2-6}
1 & 4 & 3 & 2 & 1 & 0  \\ \cline{2-6}
2 & 3 & 2 & 1 & 0 & 4  \\ \cline{2-6}
3 & 0 & 4 & 3 & 2 & 1  \\ \cline{2-6}
4 & 2 & 1 & 0 & 4 & 3 \\ \cline{2-6}
\end{tabular}\end{center}
\caption{\label{figls5}The Latin square $X_3(X_1,X_2)$ for Example \ref{ex55}.}
\end{figure}

Here the minors are all zero and therefore no permutation on $X_3$ is needed. Now, from the column beginning with $\omega_0$ we read the permutation of $X_1$. It is easy to see that it is $X_1^3$ with the switch permutation $(\omega_2,\omega_4)$. From the row beginning with $\omega_0$ we obtain the permutation $X_2^4$. Thus we have $X_3=\omega_1 \pi_1(X_1)^3X_2^4$, or equivalently $\pi_1(X_1)^2X_2X_3=\omega_1$. Notice that this constant term $\omega_1$ can be easily recovered from the Latin square above, since it is the symbol in the upper-left position, where $\pi_1(X_1)=X_2=\omega_0$ and therefore $X_3$ is equal to the constant term on the right side hand of the defining equation.

As no other defining equations with $3$ factors can be obtained, we move to the interactions of order $4$. There are few checks to do at this stage, because there are only $5$ subsets with $4$ factors and two of them are impossible: $\{X_1,X_2,X_3,X_4\}$  and $\{X_1,X_2,X_3,X_4\}$ can be excluded as they contain $\{X_1,X_2,X_3\}$, see Remark \ref{rem-sempl}. Consider the $4$-tuple $\{X_2,X_3,X_4,X_5\}$. We look at the layers defined by $X_5$ and we obtain the five Latin squares in Figure \ref{figls6}.
\begin{figure}
\begin{center}
\begin{tabular}{ccc}
\begin{tabular}{c|c|c|c|c|c|}
  \multicolumn{1} {c} { }    &
  \multicolumn{1} {c} {0} &  \multicolumn{1} {c} {1} &  \multicolumn{1} {c} {2} &  \multicolumn{1} {c} {3} &  \multicolumn{1} {c} {4}  \\
  \cline{2-6}
0 &   1 &   4  &  2 &   0 &   3  \\  \cline{2-6}
1 &   3 &   1 &   4  &  2 &   0 \\  \cline{2-6}
2 &   0 &   3 &   1  &  4 &   2 \\  \cline{2-6}
3 &   2 &   0 &  3  &  1 &   4 \\  \cline{2-6}
4 &  4  &  2 &   0  &  3  &  1 \\  \cline{2-6}
\end{tabular} &
\begin{tabular}{c|c|c|c|c|c|}
  \multicolumn{1} {c} { }    &
  \multicolumn{1} {c} {0} &  \multicolumn{1} {c} {1} &  \multicolumn{1} {c} {2} &  \multicolumn{1} {c} {3} &  \multicolumn{1} {c} {4}  \\
  \cline{2-6}
0 &   3 &   1 &   4 &   2 &   0  \\  \cline{2-6}
1 &   0  &  3  &  1 &   4 &   2 \\  \cline{2-6}
2 &   2 &   0  &  3 &   1 &   4 \\  \cline{2-6}
3 &   4 &   2 &   0 &   3 &   1 \\  \cline{2-6}
4 &   1  &  4 &   2 &   0 &   3   \\  \cline{2-6}
\end{tabular} &
\begin{tabular}{c|c|c|c|c|c|}
  \multicolumn{1} {c} { }    &
  \multicolumn{1} {c} {0} &  \multicolumn{1} {c} {1} &  \multicolumn{1} {c} {2} &  \multicolumn{1} {c} {3} &  \multicolumn{1} {c} {4}  \\
  \cline{2-6}
0 &  2 &   0 &   3  &  1 &   4 \\  \cline{2-6}
1 &   4 &   2 &   0 &   3 &   1 \\  \cline{2-6}
2 &   1  &  4  &  2  &  0 &   3 \\  \cline{2-6}
3 &   3  &  1  &  4  &  2 &   0 \\  \cline{2-6}
4 &   0 &   3  &  1  &  4  &  2   \\  \cline{2-6}
\end{tabular} \\
 $X_5=\omega_0$ &  $X_5=\omega_1$ &  $X_5=\omega_2$ \\
 \end{tabular}

\medskip
\begin{tabular}{ccc}
\begin{tabular}{c|c|c|c|c|c|}
  \multicolumn{1} {c} { }    &
  \multicolumn{1} {c} {0} &  \multicolumn{1} {c} {1} &  \multicolumn{1} {c} {2} &  \multicolumn{1} {c} {3} &  \multicolumn{1} {c} {4}  \\
  \cline{2-6}
0 &  0 &   3 &   1 &   4  &  2 \\  \cline{2-6}
1  &  2 &   0 &   3 &   1  &  4 \\  \cline{2-6}
2  &  4  &  2 &  0  &  3  &  1 \\  \cline{2-6}
3  &  1 &  4  &  2 &   0  &  3 \\  \cline{2-6}
4 &   3 &  1  &  4 &   2 &   0 \\  \cline{2-6}
\end{tabular} &
\begin{tabular}{c|c|c|c|c|c|}
  \multicolumn{1} {c} { }    &
  \multicolumn{1} {c} {0} &  \multicolumn{1} {c} {1} &  \multicolumn{1} {c} {2} &  \multicolumn{1} {c} {3} &  \multicolumn{1} {c} {4}  \\
  \cline{2-6}
0 &  4 &   2 &   0 &   3 &   1  \\  \cline{2-6}
1 &   1 &   4 &   2 &   0  &  3  \\  \cline{2-6}
2  &  3 &   1 &   4  &  2 &   0  \\  \cline{2-6}
3  &  0 &   3 &   1  &  4 &   2  \\  \cline{2-6}
4  &  2 &   0 &   3  &  1 &  4  \\  \cline{2-6}
\end{tabular} \\
 $X_5=\omega_3$ &  $X_5=\omega_4$ \\
 \end{tabular}
 \end{center}
\caption{\label{figls6}The layers $X_4(X_2,X_3)$ given $X_5$ for Example \ref{ex55}.}
\end{figure}

In all these Latin squares the $2 \times 2$ minors are all zero, and thus no permutation is needed on $X_4$. Looking at the first Latin square, one reads the permutation $(\omega_2,\omega_0,\omega_3,\omega_1,\omega_4)$ for $X_2$, i.e., $\omega_1X_2^3$, and the permutation $(\omega_1,\omega_3,\omega_0,\omega_2,\omega_4)$ for $X_3$, i.e., $\omega_2X_3^2$. Then the equation is $\pi_5(X_5)X_4=\omega_3 X_2^2X_3^3$, and this equation holds in all the five Latin squares in Figure \ref{figls6}. We only need to find the possible permutation $\pi_5$ on $X_5$. This is accomplished by looking at the constant terms in the upper-left cell. From the five Latin squares we read
\begin{multline*}
\pi_5(\omega_0)\omega_1 = \omega_3 \quad \pi_5(\omega_1)\omega_3 = \omega_3 \quad \pi_5(\omega_2)\omega_2 = \omega_3 \quad
 \\
\pi_5(\omega_3)\omega_0 = \omega_3 \quad \pi_5(\omega_4)\omega_4 = \omega_3
\end{multline*}
and these equations are satisfied when $\pi_5=(\omega_2,\omega_0,\omega_1,\omega_3,\omega_4)$. Therefore, applying $\pi_5$ to $X_5$ we have the defining equation $X_2^3X_3^2X_4\pi_5(X_5)=\omega_3$. Finally, we check that this defining equation corresponds to the second equation used in Eq. \eqref{def5a5} to define the array. Indeed, from $\pi_1(X_1)^2X_2X_3=\omega_1$ we have $X_3=\omega_1\pi_1(X_1)^3X_2$, and replacing this expression of $X_3$ into the previous equation one obtains immediately the defining equation $\pi_1(X_1)X_2X_4\pi_5(X_5) = \omega_1$.
\end{example}

\section{Conclusions} \label{future-sect}

In this paper we addressed the problem of level permutations for qualitative factors. In particular we presented two tools to check if a fraction of a $s^m$ factorial design is isomorphic or not to a regular fraction by  permutations of factor levels.  Such a problem is very important in the applications because of the special property of not partial confounding that have the regular fractions. In this framework, the coding of levels by the $s$-th roots of the unity and some tools of algebraic statistics are essential.

Future works will concern the case of designs with non prime number of levels and mixed designs, where several properties of the roots of the unity do not hold, and therefore for this class of designs a different approach must be implemented.

Moreover we want to deepen and better define the concept of mean aberration, already introduced in
\cite{fontana|rapallo|rogantin:16}. In particular, we want to limit the mean only to permutations compatible with the
design matrix of the fraction under investigation. In fact, the aberrations are calculated through the level counts,
and they are connected to each other by a convolution formula presented in the aforementioned article. This new
definition could allow us to clarify which aberrations are compatible with those of permuted regular fractions.

Finally, in order to generalize the algorithms presented in this paper for specific examples, we want to provide efficient
packages, both in symbolic and statistical software (i.e., CoCoA and R respectively), to make actual computations regarding regularity and isomorphism checks under level permutations in a general setting.

\section*{Appendix A. Basic facts in Computational Commutative Algebra}

In this appendix we collect some basic notions of Computational Commutative Algebra used in the paper. The structure of this appendix is taken from \cite{gibilisco|etal:10} (Section $1.7$ with Roberto Notari). Introductory expositions of the subject can be found in, e.g, \cite{cox|little|oshea:15} and \cite{kreuzer|robbiano:08}.

Let ${\mathbb K}$ be a numeric field. In our applications we consider ${\mathbb K}={\mathbb C}$ or ${\mathbb K}={\mathbb R}$ or ${\mathbb K}={\mathbb Z}_p$, the finite field with $p$ elements ($p$ prime). Let $ R ={\mathbb K}[x_1, \dots, x_m] $ be the ring of the polynomials in the variables $ x_1, \dots,x_m$ and with coefficients in ${\mathbb K}$. The ring operations in $R$ are the usual sum and product of polynomials.

\begin{definition} A subset $ I \subset R $ is an {\em ideal} if $ f + g
\in I $ for all $ f, g \in I $ and $ f g \in I $ for all $ f \in I
$ and all $ g \in R$.
\end{definition}

\begin{proposition} Let $ f_1, \dots, f_r \in R $.
The set $ \langle f_1, \dots, f_r \rangle = \{ f_1 g_1 + \dots +
f_r g_r : g_1, \dots, g_r \in R \} $ is the smallest ideal in
$ R $ with respect to the inclusion that contains $ f_1, \dots,
f_r$. The ideal $ \langle f_1, \dots, f_r \rangle$ is called the {\em ideal generated by} $ f_1, \dots, f_r$.
\end{proposition}

A key theorem in the theory of ideals in a polynomial ring is Hilbert's basis theorem, which states that every ideal in $R$ is finitely generated.
\begin{theorem} Given an ideal $ I \subset R,$ there exist $ f_1,
\dots, f_r \in I $ such that $ I = \langle f_1, \dots, f_r \rangle$.
\end{theorem}

The intersection of two ideals is an ideal. The union of two ideals is not an ideal in general, but the following definition can be stated.

\begin{definition} Let $ I, J \subset R $ be ideals. Then, $$ I +
J = \{ f + g : f \in I, g \in J \} $$ is the smallest ideal in
$ R $ with respect to inclusion that contains both $ I $ and $ J,$
and it is called the sum of $ I $ and $ J$.
\end{definition}

Given an ideal $I$, an equivalence relation is naturally defined on $R$: $f \sim_I g $ if $ f - g \in I $ for $ f, g \in R$. This relation is compatible with the ring operations: if $ f_1 \sim_I f_2, g_1 \sim_I g_2 $ then $ f_1 + g_1
\sim_I f_2 + g_2 $ and $ f_1 g_1 \sim_I f_2 g_2$. This fact implies that the quotient space $R/I$ can be defined and it is a ring with the operations inherited from $R$.

The definitions below are the starting point of the computational side of Commutative Algebra.

\begin{definition} A {\em term} in $ R $ is $ x^a = x_1^{a_1} \dots
x_m^{a_m} $ for $ a = (a_1, \dots, a_m) \in {\mathbb N}^m$. The set of terms is denoted with $ {\mathbb{T} }^m$.
\end{definition}

\begin{definition} A term-ordering is a {\em well ordering} $ \preccurlyeq $
on $ {\mathbb{T} }^m $ such that $ 1 \preccurlyeq x^a $ for every $ x^a \in
{\mathbb{T} }^m $ and $ x^a \preccurlyeq x^b $ implies $ x^a x^c
\preccurlyeq x^b x^c $ for every $ x^c \in {\mathbb{T} }^m$.
\end{definition}

A polynomial in $ R $ is a linear combination of a finite set of terms in $ {\mathbb{T} }^m:$ $ f = \sum_{a \in A} c_a
x^a $ where $ A $ is a finite subset of ${\mathbb N}^m$.

\begin{definition} Let $ f \in R $ be a polynomial,   $ A $
 the finite set formed by the terms in $ f$ and
 $ x^b = \max \{ x^a : a \in A \}$.  Let $ I \subset R $ be an ideal.
\begin{enumerate}
\item The term
$\operatorname{LT}(f) = c_b x^b $ is called the {\em leading term} of $ f$.
\item  The ideal generated by $\operatorname{LT}(f) $ for every $ f
\in I$ is called the {\em order ideal} of $I$ and is indicated as $\operatorname{LT(I)}$.
\end{enumerate}
\end{definition}

\begin{definition} Let $ I \subset R $ be an ideal and let $ f_1,
\dots, f_t \in I$. The set $ \{ f_1, \dots, f_t \} $ is a {\em Gr\"obner basis} of $ I $ with respect to $ \preccurlyeq
$ if
\[ \operatorname{LT}(I) = \langle \operatorname{LT}(f_1), \dots,\operatorname{LT}(f_t) \rangle.
\]
\end{definition}

Gr\"obner bases are special sets of generators for ideals in $R$ with several applications. We list here only the results used in the paper.

\begin{proposition} Let $ I \subseteq R $ be an ideal. Then, $ I =
R $ if, and only if, $ 1 \in {\mathcal {F} },$ where $ {\mathcal {F} } $ is a Gr\"obner
basis of $ I,$ with respect to any term-ordering $ \preccurlyeq $.
\end{proposition}

The theorem below is known as the {\em finiteness theorem} and is used to determine whether a system of polynomial equations has a finite number of solutions. Here $V(I)$ is the variety defined by $I$, i.e., the set of the affine points $(x_1, \ldots, x_m)$ in ${\mathbb K}^m$ such that $f(x_1, \ldots x_m)=0$ for all $f \in I$.

\begin{theorem} \label{0dimcriterion:prop}
Let $I \subset R={\mathbb K}[x_1, \ldots, x_m]$ be an ideal and fix a term-ordering $ \preccurlyeq $ on $R$. If ${\mathbb K}$ is algebraically closed, then the following five statements are equivalent:
\begin{itemize}
\item[(i)] For each $i=1,\ldots m$, there is some $c_i>0$ such that $x_i^{c_i} \in LT(I)$.

\item[(ii)] Let $G$ be a Gr\"obner basis of $I$ with respect to $ \preccurlyeq $. Then for each $i=1, \ldots, m$, there is some $c_i>0$ such that $x_i^{c_i}= LT(f)$ for some $f \in G$.

\item[(iii)] The set $\{x^{\alpha} : x^{\alpha} \notin LT(I)\}$ is finite.

\item[(iv)] The set $R/I$ is finite-dimensional as ${\mathbb K}$-vector space.

\item[(v)] $V(I)$ is a finite set.
\end{itemize}
\end{theorem}

\begin{definition} Let $ I \subset R $ be an ideal. A polynomial $
f = \sum_{a \in A} c_a x^a $ is in {\em normal form} with respect to $
\preccurlyeq $ and $ I $ if $ x^a \notin \operatorname{LT}(I) $ for each $ a \in
A$.
\end{definition}

\begin{proposition} Let $ I \subset R $ be an ideal. For every $ f
\in R $ there exists a unique polynomial, indicated as $\operatorname{NF}(f) \in R $, in normal form with
respect to $ \preccurlyeq $ and $ I $ such that $ f - \operatorname{NF}(f) \in
I.$
\end{proposition}

The normal form of a polynomial $f$ can be computed from $f$ and a Gr\"obner basis of $I$ with respect to $ \preccurlyeq$. The normal form solves the problem of ideal membership: a polynomial $f$ belongs to the ideal $I$ if and only if $\operatorname{NF}(f)=0$.

\begin{definition}
Given an ideal  $I \subset R={\mathbb K}[x_1, \ldots, x_m]$, the $l$-th {\em elimination ideal} $I_l$ is
the ideal of ${\mathbb K}[x_{l+1},\ldots, x_m]$ defined by
\[
I_l = I \cap  {\mathbb K} [x_{l+1}, \ldots, x_m] .
\]
\end{definition}

The elimination ideal is the generalization of the Gaussian elimination to the polynomial case. Computing a Gr\"obner basis of $I$ with respect to a special term-ordering, we can actually compute the ideal $I_l$. The {\em lex term-ordering} on ${\mathbb T}^m$ is defined by: for $a \ne b$, $x^a  \preccurlyeq x^b$ if and only if the leftmost nonzero entry of the vector difference $b-a \in {\mathbb Z}^m$ is positive.

\begin{theorem}
Let $I \subset R={\mathbb K}[x_1, \ldots, x_m]$ be an ideal and let $G$ be a Gr\"obner basis of $I$ with respect to lex term-ordering with $x_1 > x_2 > \cdots > x-n$. Then, for every $l=0, \ldots ,m$, the set
\[
G_l = G \cap {\mathbb K}[x_{l+1}, \ldots , x_m]
\]
is a Gr\"obner basis of the elimination ideal $I_l$.
\end{theorem}

\newpage

\section*{Appendix B. Indicator functions of fractions used in the paper}

\noindent Fraction $A$ in Example 2, first fraction in \cite{eendebak:sito} for the $5^3$ design:
 {\footnotesize{
 \begin{multline*}
 F_A=\frac 1 5 (
   1 + (5 \omega_3 +5 \omega_2 +5 \omega_1 -2)    X_1^4  X_2^4  X_3^4+
  (-5 \omega_3 -5 \omega_2 -5 \omega_1 -9)    X_1^2  X_2^4  X_3^4+ \\
  (-4 \omega_1 +1)    X_1  X_2^4  X_3^4+
  (5 \omega_3 +5 \omega_1 -3)    X_1^4  X_2^3  X_3^4+
  (-5 \omega_3 -5 \omega_1 -8)    X_1^3  X_2^3  X_3^4+ \\
  (-4 \omega_1 +1)    X_1  X_2^3  X_3^4+
  (-5 \omega_3 -5 \omega_2 -8 \omega_1 -3)    X_1^4  X_2^2  X_3^4+
  (5 \omega_3 -4)    X_1^3  X_2^2  X_3^4+ \\
  (5 \omega_2 +5 \omega_1 -3)    X_1^2  X_2^2  X_3^4+
  (-5 \omega_3 -7 \omega_1 -2)    X_1^4  X_2  X_3^4+
  (-3 \omega_2 +2 \omega_1 +2)    X_1^3  X_2  X_3^4+\\
  (5 \omega_3 -3 \omega_1 +2)    X_1^2  X_2  X_3^4+
  (-2 \omega_1 +13)    X_1  X_2  X_3^4+
  (-5 \omega_3 -5 \omega_1 -8)    X_1^4  X_2^4  X_3^3+\\
  (5 \omega_3 -3 \omega_1 +2)    X_1^3  X_2^4  X_3^3+
  (5 \omega_1 -4)    X_1  X_2^4  X_3^3+
  (5 \omega_3 -4)    X_1^4  X_2^3  X_3^3+\\
  (-5 \omega_3 -7)    X_1^3  X_2^3  X_3^3+
  (-4 \omega_1 +1)    X_1^2  X_2^3  X_3^3+
  (-3 \omega_2 +2 \omega_1 +2)    X_1^4  X_2^2  X_3^3+\\
  (-5 \omega_2 -7 \omega_1 -2)    X_1^3  X_2^2  X_3^3+
  (-2 \omega_1 +13)    X_1^2  X_2^2  X_3^3+
  (5 \omega_2 -3 \omega_1 +2)    X_1  X_2^2  X_3^3+\\
  (5 \omega_2 +5 \omega_1 -3)    X_1^3  X_2  X_3^3+
  (-4 \omega_1 +1)    X_1^2  X_2  X_3^3+
  (-5 \omega_2 -5 \omega_1 -8)    X_1  X_2  X_3^3+\\
  (5 \omega_2 +5 \omega_1 -3)    X_1^4  X_2^4  X_3^2+
  (-4 \omega_1 +1)    X_1^3  X_2^4  X_3^2+
  (-5 \omega_2 -5 \omega_1 -8)    X_1^2  X_2^4  X_3^2+\\
  (5 \omega_3 -3 \omega_1 +2)    X_1^4  X_2^3  X_3^2+
  (-2 \omega_1 +13)    X_1^3  X_2^3  X_3^2+
  (-2 \omega_3 +3 \omega_2 +3 \omega_1 +3)    X_1^2  X_2^3  X_3^2+ \\
  (-5 \omega_3 -5 \omega_2 -8 \omega_1 -3)    X_1  X_2^3  X_3^2+
  (-4 \omega_1 +1)    X_1^3  X_2^2  X_3^2+
  (-5 \omega_2 -7)    X_1^2  X_2^2  X_3^2+\\
  (5 \omega_2 -4)    X_1  X_2^2  X_3^2+
  (-5 \omega_3 -5 \omega_2 -5 \omega_1 -9)    X_1^4  X_2  X_3^2+
  (5 \omega_2 -3 \omega_1 +2)    X_1^2  X_2  X_3^2+ \\
  (5 \omega_3 +5 \omega_1 -3)    X_1  X_2  X_3^2+
  (-2 \omega_1 +13)    X_1^4  X_2^4  X_3+
  (5 \omega_2 -3 \omega_1 +2)    X_1^3  X_2^4  X_3+\\
  (-5 \omega_3 -5 \omega_2 -8 \omega_1 -3)    X_1^2  X_2^4  X_3+
  (5 \omega_3 -2 \omega_2 +3 \omega_1 +3)    X_1  X_2^4  X_3+
  (-5 \omega_2 -5 \omega_1 -8)    X_1^3  X_2^3  X_3+\\
  (5 \omega_2 -4)    X_1^2  X_2^3  X_3+
  (-3 \omega_2 +2 \omega_1 +2)    X_1  X_2^3  X_3+
  (-4 \omega_1 +1)    X_1^4  X_2^2  X_3+\\
  (5 \omega_3 +5 \omega_1 -3)    X_1^2  X_2^2  X_3+
  (-5 \omega_3 -5 \omega_1 -8)    X_1  X_2^2  X_3+
  (-4 \omega_1 +1)    X_1^4  X_2  X_3+\\
  (5 \omega_1 -4)    X_1^3  X_2  X_3+
  (-5 \omega_1 -7)    X_1  X_2  X_3)
 \end{multline*}}}

\noindent Fraction B in Example 2, second fraction in \cite{eendebak:sito} for the $5^3$ design:  \\
{\footnotesize{
\begin{equation*}
F_B=\frac 1 5 \left(1+ X_1 X_2 X_3^4 +X_1^2 X_2^2 X_3^3+X_1^3 X_2^3 X_3^2+ X_1^4 X_2^4 X_3\right)= \frac 1 5 \sum_{i=0}^4 \left( X_1 X_2 X_3^4  \right)^{i}
\end{equation*}
}}

\noindent Fraction C in Example 2, other fraction for the $5^3$ design: \\
{\footnotesize{
 \begin{multline*}
 F_C=\frac 1 {25} \left(
5 +(  \omega_4 +\omega_3 +3\omega_1)   X_1 X_2^4 X_3^4+(  \omega_4 +2\omega_2 +2\omega_1)   X_1 X_2^3 X_3^4+\right.\\
(  2\omega_3 +\omega_2 +2\omega_1)   X_1 X_2^2 X_3^4+
(  \omega_4 +\omega_1 +3)   X_1 X_2 X_3^4+(  \omega_4 +2\omega_2 +2\omega_1)   X_1^2 X_2^4 X_3^3+\\
(  \omega_3 +3\omega_2 +\omega_1)   X_1^2 X_2^3 X_3^3+
(  \omega_3 +\omega_2 +3)   X_1^2 X_2^2 X_3^3+(  2\omega_4 +\omega_3 +2\omega_2)   X_1^2 X_2 X_3^3+\\
(  2\omega_3 +\omega_2 +2\omega_1)   X_1^3 X_2^4 X_3^2+(  \omega_3 +\omega_2 +3)   X_1^3 X_2^3 X_3^2+
(  \omega_4 +3\omega_3 +\omega_2)   X_1^3 X_2^2 X_3^2+\\
(  2\omega_4 +2\omega_3 +\omega_1)   X_1^3 X_2 X_3^2+
(  \omega_4 +\omega_1 +3)   X_1^4 X_2^4 X_3+(  2\omega_4 +\omega_3 +2\omega_2)   X_1^4 X_2^3 X_3+\\
\left .(  2\omega_4 +2\omega_3 +\omega_1)   X_1^4 X_2^2 X_3+(  3\omega_4 +\omega_2 +\omega_1   X_1^4 X_2 X_3 \right)
\end{multline*}}}

\noindent Fraction D, fraction for the $5^5$ design in Example 5:
\\
{\footnotesize{
 \begin{multline*}
 F_D=\frac 1 {125} \left(
5+\omega_2X_1^4X_2^4X_3^4+(\omega_3 +\omega_2)X_1^3X_2^4X_3^4+(\omega_3 +\omega_1 +1)X_1^2X_2^4X_3^4+  \right. \\
(\omega_4 +\omega_3 +\omega_2 +1)X_1X_2^4X_3^4+(\omega_2 +2\omega_1 +1)X_1^4X_2^3X_3^4+(\omega_4 +\omega_3 +1)X_1^3X_2^3X_3^4+\\
(\omega_4 +1)X_1^2X_2^3X_3^4+X_1X_2^3X_3^4+(\omega_3 +\omega_2 +2)X_1^4X_2^2X_3^4+\\
(\omega_3 +\omega_1 +1)X_1^3X_2^2X_3^4+(\omega_2 +1)X_1^2X_2^2X_3^4+\omega_4X_1X_2^2X_3^4+\\
\omega_4X_1^4X_2X_3^4+(\omega_3 +\omega_1)X_1^3X_2X_3^4+(\omega_4 +\omega_1 +1)X_1^2X_2X_3^4+ \\
(\omega_4 +\omega_2 +\omega_1 +1)X_1X_2X_3^4+(\omega_4 +1)X_1^4X_2^4X_3^3+(\omega_4 +\omega_1 +2)X_1^3X_2^4X_3^3+\\
\omega_3X_1^2X_2^4X_3^3+(\omega_2 +\omega_1 +1)X_1X_2^4X_3^3+(\omega_2 +\omega_1 +1)X_1^4X_2^3X_3^3+ \\
\omega_4X_1^3X_2^3X_3^3+(\omega_4 +\omega_3 +\omega_1 +1)X_1^2X_2^3X_3^3+(\omega_4 +\omega_1)X_1X_2^3X_3^3+\\
(\omega_3 +\omega_2 +1)X_1^4X_2^2X_3^3+\omega_3X_1^3X_2^2X_3^3+(\omega_4 +\omega_3 +\omega_2 +1)X_1^2X_2^2X_3^3+ \\
(\omega_2 +\omega_1)X_1X_2^2X_3^3+(\omega_3 +1)X_1^4X_2X_3^3+(\omega_4 +2\omega_2 +1)X_1^3X_2X_3^3+\\
X_1^2X_2X_3^3+(\omega_3 +\omega_1 +1)X_1X_2X_3^3+(\omega_4 +\omega_2 +1)X_1^4X_2^4X_3^2+X_1^3X_2^4X_3^2+\\
(2\omega_3 +\omega_1 +1)X_1^2X_2^4X_3^2+(\omega_2 +1)X_1X_2^4X_3^2+(\omega_4 +\omega_3)X_1^4X_2^3X_3^2+\\
(\omega_3 +\omega_2 +\omega_1 +1)X_1^3X_2^3X_3^2+\omega_2X_1^2X_2^3X_3^2+(\omega_3 +\omega_2 +1)X_1X_2^3X_3^2+\\
(\omega_4 +\omega_1)X_1^4X_2^2X_3^2+(\omega_4 +\omega_2 +\omega_1 +1)X_1^3X_2^2X_3^2+ \omega_1X_1^2X_2^2X_3^2+\\
(\omega_4 +\omega_3 +1)X_1X_2^2X_3^2+(\omega_4 +\omega_3 +1)X_1^4X_2X_3^2+  \omega_2X_1^3X_2X_3^2+\\
(\omega_4 +\omega_1 +2)X_1^2X_2X_3^2+(\omega_1 +1)X_1X_2X_3^2+ (\omega_4 +\omega_3 +\omega_1 +1)X_1^4X_2^4X_3+\\
(\omega_4 +\omega_1 +1)X_1^3X_2^4X_3+(\omega_4 +\omega_2)X_1^2X_2^4X_3+\omega_1X_1X_2^4X_3+\\
\omega_1X_1^4X_2^3X_3+(\omega_3 +1)X_1^3X_2^3X_3+(\omega_4 +\omega_2 +1)X_1^2X_2^3X_3+ \\
(\omega_3 +\omega_2 +2)X_1X_2^3X_3+X_1^4X_2^2X_3+(\omega_1 +1)X_1^3X_2^2X_3+ \\
(\omega_2 +\omega_1 +1)X_1^2X_2^2X_3+(2\omega_4 +\omega_3 +1)X_1X_2^2X_3+(\omega_3 +\omega_2 +\omega_1 +1)X_1^4X_2X_3+ \\
(\omega_4 +\omega_2 +1)X_1^3X_2X_3+(\omega_3 +\omega_2)X_1^2X_2X_3+\omega_3X_1X_2X_3+\\
(\omega_4 +\omega_3 +3\omega_1)X_1^4X_2^4X_4^4X_5^4+(\omega_4 +2\omega_2 +2\omega_1)X_1^3X_2^4X_4^4X_5^4+\\
(2\omega_3 +\omega_2 +2\omega_1)X_1^2X_2^4X_4^4X_5^4+ (\omega_4 +\omega_1 +3)X_1X_2^4X_4^4X_5^4+\\
(\omega_3 +\omega_2 +\omega_1)X_1^4X_3^4X_4^4X_5^4+X_1^3X_3^4X_4^4X_5^4+(\omega_4 +\omega_2 +\omega_1 +1)X_1^2X_3^4X_4^4X_5^4+\\
(\omega_2 +1)X_1X_3^4X_4^4X_5^4+ (\omega_4 +\omega_3 +\omega_2 +\omega_1)X_1^4X_2^3X_3^4X_4^4X_5^4+\\
(\omega_3 +\omega_1 +1)X_1^3X_2^3X_3^4X_4^4X_5^4+(\omega_4 +\omega_3)X_1^2X_2^3X_3^4X_4^4X_5^4+\omega_4X_1X_2^3X_3^4X_4^4X_5^4+\\
(\omega_2 +3\omega_1 +1)X_2^2X_3^4X_4^4X_5^4+(\omega_3 +\omega_1 +1)X_1^4X_2X_3^4X_4^4X_5^4+\omega_1X_1^3X_2X_3^4X_4^4X_5^4+\\
(2\omega_4 +\omega_2 +\omega_1)X_1^2X_2X_3^4X_4^4X_5^4+(\omega_3 +\omega_1)X_1X_2X_3^4X_4^4X_5^4+ \omega_2X_1^4X_3^3X_4^4X_5^4+\\
(\omega_4 +\omega_1)X_1^3X_3^3X_4^4X_5^4+(\omega_3 +\omega_1 +1)X_1^2X_3^3X_4^4X_5^4+ (\omega_4 +\omega_3 +2\omega_1)X_1X_3^3X_4^4X_5^4+\\
(\omega_4 +\omega_1 +1)X_1^4X_2^3X_3^3X_4^4X_5^4+\omega_3X_1^3X_2^3X_3^3X_4^4X_5^4+(\omega_2 +2\omega_1 +1)X_1^2X_2^3X_3^3X_4^4X_5^4+\\
(\omega_2 +\omega_1)X_1X_2^3X_3^3X_4^4X_5^4+(2\omega_2 +\omega_1 +2)X_2^2X_3^3X_4^4X_5^4+ \omega_3X_1^4X_2X_3^3X_4^4X_5^4+\\
(\omega_4 +\omega_3)X_1^3X_2X_3^3X_4^4X_5^4+(\omega_4 +\omega_2 +\omega_1)X_1^2X_2X_3^3X_4^4X_5^4+ \\
(\omega_4 +\omega_3 +\omega_1 +1)X_1X_2X_3^3X_4^4X_5^4+(\omega_3 +2\omega_2 +\omega_1)X_1^4X_3^2X_4^4X_5^4+\\
(\omega_4 +\omega_1 +1)X_1^3X_3^2X_4^4X_5^4+ (\omega_1 +1)X_1^2X_3^2X_4^4X_5^4+\omega_1X_1X_3^2X_4^4X_5^4+\\
(\omega_4 +\omega_1)X_1^4X_2^3X_3^2X_4^4X_5^4+(2\omega_3 +\omega_1 +1)X_1^3X_2^3X_3^2X_4^4X_5^4+\omega_1X_1^2X_2^3X_3^2X_4^4X_5^4+\\
(\omega_4 +\omega_2 +\omega_1)X_1X_2^3X_3^2X_4^4X_5^4+(2\omega_4 +2\omega_3 +\omega_1)X_2^2X_3^2X_4^4X_5^4+\\
(\omega_4 +\omega_2 +\omega_1 +1)X_1^4X_2X_3^2X_4^4X_5^4+(\omega_2 +\omega_1 +1)X_1^3X_2X_3^2X_4^4X_5^4+ (\omega_3 +1)X_1^2X_2X_3^2X_4^4X_5^4+
\end{multline*}}}
{\footnotesize{
 \begin{multline*}
\omega_2X_1X_2X_3^2X_4^4X_5^4+(\omega_4 +1)X_1^4X_3X_4^4X_5^4+ (\omega_4 +\omega_3 +\omega_2 +\omega_1)X_1^3X_3X_4^4X_5^4+\\
\omega_3X_1^2X_3X_4^4X_5^4+(\omega_4 +\omega_3 +\omega_1)X_1X_3X_4^4X_5^4+ X_1^4X_2^3X_3X_4^4X_5^4+\\
(\omega_4 +\omega_2)X_1^3X_2^3X_3X_4^4X_5^4+(\omega_2 +\omega_1 +1)X_1^2X_2^3X_3X_4^4X_5^4+(\omega_3 +\omega_2 +\omega_1 +1)X_1X_2^3X_3X_4^4X_5^4+\\
(\omega_4 +\omega_3 +3\omega_1)X_2^2X_3X_4^4X_5^4+(\omega_1 +1)X_1^4X_2X_3X_4^4X_5^4+ (\omega_2 +2\omega_1 +1)X_1^3X_2X_3X_4^4X_5^4+\\
\omega_4X_1^2X_2X_3X_4^4X_5^4+(\omega_3 +\omega_2 +\omega_1)X_1X_2X_3X_4^4X_5^4+ (\omega_4 +2\omega_2 +2\omega_1)X_1^4X_2^3X_4^3X_5^3+\\
(\omega_3 +3\omega_2 +\omega_1)X_1^3X_2^3X_4^3X_5^3+(\omega_3 +\omega_2 +3)X_1^2X_2^3X_4^3X_5^3+\\
(2\omega_4 +\omega_3 +2\omega_2)X_1X_2^3X_4^3X_5^3+(\omega_2 +1)X_1^4X_3^4X_4^3X_5^3+(2\omega_4 +\omega_2 +\omega_1)X_1^3X_3^4X_4^3X_5^3+\\
\omega_2X_1^2X_3^4X_4^3X_5^3+(\omega_3 +\omega_2 +1)X_1X_3^4X_4^3X_5^3+(2\omega_3 +\omega_2 +2\omega_1)X_2^4X_3^4X_4^3X_5^3+\\
(\omega_1 +1)X_1^4X_2^2X_3^4X_4^3X_5^3+ (\omega_4 +\omega_3 +\omega_2 +1)X_1^3X_2^2X_3^4X_4^3X_5^3+\omega_4X_1^2X_2^2X_3^4X_4^3X_5^3+\\
(\omega_4 +\omega_2 +1)X_1X_2^2X_3^4X_4^3X_5^3+ \omega_2X_1^4X_2X_3^4X_4^3X_5^3+(\omega_3 +\omega_2)X_1^3X_2X_3^4X_4^3X_5^3+\\
(\omega_4 +\omega_3 +\omega_2)X_1^2X_2X_3^4X_4^3X_5^3+(\omega_2 +2\omega_1 +1)X_1X_2X_3^4X_4^3X_5^3+\\
(\omega_4 +\omega_3 +\omega_2 +1)X_1^4X_3^3X_4^3X_5^3+(\omega_4 +\omega_2 +\omega_1)X_1^3X_3^3X_4^3X_5^3+\\
(2\omega_3 +\omega_2 +2\omega_1)X_2^4X_3^4X_4^3X_5^3+(\omega_1 +1)X_1^4X_2^2X_3^4X_4^3X_5^3+(\omega_4 +\omega_3 +\omega_2 +1)X_1^3X_2^2X_3^4X_4^3X_5^3+\\
\omega_4X_1^2X_2^2X_3^4X_4^3X_5^3+(\omega_4 +\omega_2 +1)X_1X_2^2X_3^4X_4^3X_5^3+ \omega_2X_1^4X_2X_3^4X_4^3X_5^3+\\
(\omega_3 +\omega_2)X_1^3X_2X_3^4X_4^3X_5^3+(\omega_4 +\omega_3 +\omega_2)X_1^2X_2X_3^4X_4^3X_5^3+\\
(\omega_2 +2\omega_1 +1)X_1X_2X_3^4X_4^3X_5^3+(\omega_4 +\omega_3 +\omega_2 +1)X_1^4X_3^3X_4^3X_5^3+\\
(\omega_4 +\omega_2 +\omega_1)X_1^3X_3^3X_4^3X_5^3+(\omega_4 +1)X_1^2X_3^3X_4^3X_5^3+X_1X_3^3X_4^3X_5^3+\\
(\omega_4 +3\omega_2 +1)X_2^4X_3^3X_4^3X_5^3+(\omega_4 +2\omega_3 +\omega_2)X_1^4X_2^2X_3^3X_4^3X_5^3+ \\
(\omega_2 +\omega_1 +1)X_1^3X_2^2X_3^3X_4^3X_5^3+(\omega_2 +\omega_1)X_1^2X_2^2X_3^3X_4^3X_5^3+\omega_2X_1X_2^2X_3^3X_4^3X_5^3+\\
(\omega_3 +\omega_1)X_1^4X_2X_3^3X_4^3X_5^3+(\omega_4 +\omega_3 +\omega_2 +\omega_1)X_1^3X_2X_3^3X_4^3X_5^3+\omega_3X_1^2X_2X_3^3X_4^3X_5^3+ \\
(\omega_2 +\omega_1 +1)X_1X_2X_3^3X_4^3X_5^3+\omega_1X_1^4X_3^2X_4^3X_5^3+(\omega_3 +1)X_1^3X_3^2X_4^3X_5^3+ \\
(\omega_3 +\omega_2 +\omega_1)X_1^2X_3^2X_4^3X_5^3+(\omega_4 +\omega_3 +\omega_2 +\omega_1)X_1X_3^2X_4^3X_5^3+(\omega_3 +3\omega_2 +\omega_1)X_2^4X_3^2X_4^3X_5^3+ \\
\omega_3X_1^4X_2^2X_3^2X_4^3X_5^3+(\omega_2 +1)X_1^3X_2^2X_3^2X_4^3X_5^3+(\omega_4 +\omega_2 +\omega_1)X_1^2X_2^2X_3^2X_4^3X_5^3+ \\
(\omega_4 +2\omega_2 +1)X_1X_2^2X_3^2X_4^3X_5^3+(\omega_4 +\omega_2 +1)X_1^4X_2X_3^2X_4^3X_5^3+X_1^3X_2X_3^2X_4^3X_5^3+\\
(\omega_4 +\omega_2 +\omega_1 +1)X_1^2X_2X_3^2X_4^3X_5^3+(\omega_4 +\omega_3)X_1X_2X_3^2X_4^3X_5^3+(\omega_2 +\omega_1 +1)X_1^4X_3X_4^3X_5^3+\\
\omega_4X_1^3X_3X_4^3X_5^3+(\omega_3 +2\omega_2 +\omega_1)X_1^2X_3X_4^3X_5^3+(\omega_3 +\omega_2)X_1X_3X_4^3X_5^3+\\
(2\omega_4 +\omega_2 +2)X_2^4X_3X_4^3X_5^3+(\omega_4 +\omega_3 +\omega_2)X_1^4X_2^2X_3X_4^3X_5^3+\omega_1X_1^3X_2^2X_3X_4^3X_5^3+\\
(\omega_3 +\omega_2 +\omega_1 +1)X_1^2X_2^2X_3X_4^3X_5^3+(\omega_3 +\omega_1)X_1X_2^2X_3X_4^3X_5^3+(\omega_4 +2\omega_2 +1)X_1^4X_2X_3X_4^3X_5^3+\\
(\omega_3 +\omega_2 +1)X_1^3X_2X_3X_4^3X_5^3+(\omega_4 +\omega_2)X_1^2X_2X_3X_4^3X_5^3+\omega_1X_1X_2X_3X_4^3X_5^3+\\
(2\omega_3 +\omega_2 +2\omega_1)X_1^4X_2^2X_4^2X_5^2+(\omega_3 +\omega_2 +3)X_1^3X_2^2X_4^2X_5^2+\\
(\omega_4 +3\omega_3 +\omega_2)X_1^2X_2^2X_4^2X_5^2+(2\omega_4 +2\omega_3 +\omega_1)X_1X_2^2X_4^2X_5^2+(\omega_3 +\omega_2)X_1^4X_3^4X_4^2X_5^2+\\
(\omega_4 +2\omega_3 +\omega_2)X_1^3X_3^4X_4^2X_5^2+\omega_1X_1^2X_3^4X_4^2X_5^2+(\omega_4 +\omega_3 +1)X_1X_3^4X_4^2X_5^2+ \\
\omega_4X_1^4X_2^4X_3^4X_4^2X_5^2+(\omega_3 +\omega_1)X_1^3X_2^4X_3^4X_4^2X_5^2+(\omega_3 +\omega_2 +1)X_1^2X_2^4X_3^4X_4^2X_5^2+ \\
(2\omega_3 +\omega_1 +1)X_1X_2^4X_3^4X_4^2X_5^2+(\omega_4 +\omega_2)X_1^4X_2^3X_3^4X_4^2X_5^2+(\omega_4 +\omega_3 +\omega_2 +1)X_1^3X_2^3X_3^4X_4^2X_5^2+ \\
\omega_4X_1^2X_2^3X_3^4X_4^2X_5^2+(\omega_3 +\omega_2 +\omega_1)X_1X_2^3X_3^4X_4^2X_5^2+(\omega_3 +2\omega_1 +2)X_2X_3^4X_4^2X_5^2+\\
(\omega_4 +\omega_3 +\omega_2 +\omega_1)X_1^4X_3^3X_4^2X_5^2+(\omega_4 +\omega_3 +\omega_2)X_1^3X_3^3X_4^2X_5^2+(\omega_2 +1)X_1^2X_3^3X_4^2X_5^2+ \\
\omega_4X_1X_3^3X_4^2X_5^2+(\omega_2 +\omega_1)X_1^4X_2^4X_3^3X_4^2X_5^2+(\omega_4 +\omega_3 +\omega_1 +1)X_1^3X_2^4X_3^3X_4^2X_5^2+\\
X_1^2X_2^4X_3^3X_4^2X_5^2+(\omega_3 +\omega_1 +1)X_1X_2^4X_3^3X_4^2X_5^2+(2\omega_3 +\omega_1 +1)X_1^4X_2^3X_3^3X_4^2X_5^2+\\
\end{multline*}
{\footnotesize{
 \begin{multline*}
 (\omega_4 +\omega_3 +\omega_1)X_1^3X_2^3X_3^3X_4^2X_5^2+(\omega_3 +1)X_1^2X_2^3X_3^3X_4^2X_5^2+\omega_2X_1X_2^3X_3^3X_4^2X_5^2+\\
(\omega_4 +3\omega_3 +\omega_2)X_2X_3^3X_4^2X_5^2+X_1^4X_3^2X_4^2X_5^2+(\omega_1 +1)X_1^3X_3^2X_4^2X_5^2+\\
(\omega_4 +\omega_3 +\omega_1)X_1^2X_3^2X_4^2X_5^2+(\omega_3 +\omega_2 +\omega_1 +1)X_1X_3^2X_4^2X_5^2+
(\omega_4 +\omega_3 +1)X_1^4X_2^4X_3^2X_4^2X_5^2+\\ \omega_2X_1^3X_2^4X_3^2X_4^2X_5^2+
(\omega_4 +\omega_3 +\omega_2 +\omega_1)X_1^2X_2^4X_3^2X_4^2X_5^2+(\omega_4 +\omega_2)X_1X_2^4X_3^2X_4^2X_5^2+\\
\omega_3X_1^4X_2^3X_3^2X_4^2X_5^2+(\omega_4 +\omega_3)X_1^3X_2^3X_3^2X_4^2X_5^2+(\omega_4 +\omega_3 +1)X_1^2X_2^3X_3^2X_4^2X_5^2+\\
(\omega_3 +2\omega_2 +\omega_1)X_1X_2^3X_3^2X_4^2X_5^2+(3\omega_3 +\omega_1 +1)X_2X_3^2X_4^2X_5^2+(\omega_3 +\omega_2 +1)X_1^4X_3X_4^2X_5^2+\\
\omega_3X_1^3X_3X_4^2X_5^2+(\omega_4 +\omega_3 +2\omega_1)X_1^2X_3X_4^2X_5^2+(\omega_3 +1)X_1X_3X_4^2X_5^2+\\
(2\omega_4 +\omega_3 +1)X_1^4X_2^4X_3X_4^2X_5^2+(\omega_3 +\omega_2 +\omega_1)X_1^3X_2^4X_3X_4^2X_5^2+(\omega_3 +\omega_2)X_1^2X_2^4X_3X_4^2X_5^2+\\
\omega_3X_1X_2^4X_3X_4^2X_5^2+(\omega_3 +\omega_1 +1)X_1^4X_2^3X_3X_4^2X_5^2+\omega_1X_1^3X_2^3X_3X_4^2X_5^2+\\
(\omega_3 +\omega_2 +\omega_1 +1)X_1^2X_2^3X_3X_4^2X_5^2+(\omega_4 +1)X_1X_2^3X_3X_4^2X_5^2+(2\omega_4 +\omega_3 +2\omega_2)X_2X_3X_4^2X_5^2+\\
(\omega_4 +\omega_1 +3)X_1^4X_2X_4X_5+(2\omega_4 +\omega_3 +2\omega_2)X_1^3X_2X_4X_5+(2\omega_4 +2\omega_3 +\omega_1)X_1^2X_2X_4X_5+\\
(3\omega_4 +\omega_2 +\omega_1)X_1X_2X_4X_5+(\omega_4 +\omega_2 +\omega_1)X_1^4X_3^4X_4X_5+\omega_2X_1^3X_3^4X_4X_5+\\
(\omega_4 +\omega_3 +\omega_2 +\omega_1)X_1^2X_3^4X_4X_5+(\omega_1 +1)X_1X_3^4X_4X_5+(\omega_4 +\omega_3 +\omega_2)X_1^4X_2^4X_3^4X_4X_5+\\
\omega_1X_1^3X_2^4X_3^4X_4X_5+(2\omega_4 +\omega_3 +1)X_1^2X_2^4X_3^4X_4X_5+(\omega_4 +1)X_1X_2^4X_3^4X_4X_5+\\
(3\omega_4 +\omega_2 +\omega_1)X_2^3X_3^4X_4X_5+(\omega_4 +\omega_3 +\omega_2 +1)X_1^4X_2^2X_3^4X_4X_5+\\
(\omega_4 +\omega_3 +1)X_1^3X_2^2X_3^4X_4X_5+(\omega_3 +\omega_1)X_1^2X_2^2X_3^4X_4X_5+X_1X_2^2X_3^4X_4X_5+\\
\omega_4X_1^4X_3^3X_4X_5+(\omega_4 +1)X_1^3X_3^3X_4X_5+(\omega_4 +\omega_1 +1)X_1^2X_3^3X_4X_5+(\omega_4 +2\omega_3 +\omega_2)X_1X_3^3X_4X_5+\\
\omega_3X_1^4X_2^4X_3^3X_4X_5+(\omega_2 +1)X_1^3X_2^4X_3^3X_4X_5+(\omega_4 +\omega_3 +1)X_1^2X_2^4X_3^3X_4X_5+\\
(\omega_4 +\omega_3 +\omega_1 +1)X_1X_2^4X_3^3X_4X_5+(\omega_4 +2\omega_2 +2\omega_1)X_2^3X_3^3X_4X_5+(\omega_4 +\omega_3 +\omega_1)X_1^4X_2^2X_3^3X_4X_5+\\
\omega_4X_1^3X_2^2X_3^3X_4X_5+(\omega_4 +2\omega_2 +1)X_1^2X_2^2X_3^3X_4X_5+(\omega_4 +\omega_1)X_1X_2^2X_3^3X_4X_5+\\
(2\omega_4 +\omega_2 +\omega_1)X_1^4X_3^2X_4X_5+(\omega_4 +\omega_2 +1)X_1^3X_3^2X_4X_5+(\omega_4 +\omega_1)X_1^2X_3^2X_4X_5+\\
\omega_3X_1X_3^2X_4X_5+(\omega_4 +\omega_2 +\omega_1 +1)X_1^4X_2^4X_3^2X_4X_5+(\omega_4 +\omega_3 +\\
\omega_1)X_1^3X_2^4X_3^2X_4X_5+(\omega_2 +\omega_1)X_1^2X_2^4X_3^2X_4X_5+\omega_2X_1X_2^4X_3^2X_4X_5+\\
(\omega_4 +2\omega_3 +2)X_2^3X_3^2X_4X_5+(\omega_4 +\omega_3)X_1^4X_2^2X_3^2X_4X_5+(2\omega_4 +\omega_3 +1)X_1^3X_2^2X_3^2X_4X_5+\\
\omega_2X_1^2X_2^2X_3^2X_4X_5+(\omega_4 +\omega_1 +1)X_1X_2^2X_3^2X_4X_5+(\omega_3 +1)X_1^4X_3X_4X_5+\\
(\omega_4 +\omega_3 +\omega_1 +1)X_1^3X_3X_4X_5+X_1^2X_3X_4X_5+(\omega_4 +\omega_3 +\omega_2)X_1X_3X_4X_5+\\
(\omega_4 +\omega_2)X_1^4X_2^4X_3X_4X_5+(\omega_4 +\omega_3 +2\omega_1)X_1^3X_2^4X_3X_4X_5+\omega_4X_1^2X_2^4X_3X_4X_5+\\
(\omega_4 +\omega_2 +1)X_1X_2^4X_3X_4X_5+(3\omega_4 +\omega_3 +1)X_2^3X_3X_4X_5+\omega_1X_1^4X_2^2X_3X_4X_5+\\
\left.
(\omega_2 +\omega_1)X_1^3X_2^2X_3X_4X_5+(\omega_4 +\omega_2 +1)X_1^2X_2^2X_3X_4X_5+(\omega_4 +\omega_3 +\omega_2 +\omega_1)X_1X_2^2X_3X_4X_5 \right)
\end{multline*}}}

\section*{Acknowledgements}

We thank Anna Bigatti (Universit\`a di Genova) for her valuable help in using {\tt CoCoA}. This work is partially funded by INdAM (National Institute for Higher Mathematics) through a GNAMPA-INdAM Project 2017. This research is original and has a financial support of the Universit\`a del Piemonte Orientale.

%

\end{document}